\documentclass{article}

    \usepackage{graphicx}
    
    \usepackage{caption}
    \DeclareCaptionFormat{nocaption}{}
    \usepackage{float}
    \usepackage{xcolor} 
    \usepackage{enumerate} 
    \usepackage{amsmath} 
    \usepackage{amssymb} 
    \usepackage{textcomp} 
    \AtBeginDocument{%
    }
    \usepackage{upquote} 
    \usepackage{eurosym} 

    \usepackage{iftex}
    \ifPDFTeX
        \usepackage[T1]{fontenc}
        \IfFileExists{alphabeta.sty}{
              \usepackage{alphabeta}
          }{
              \usepackage[mathletters]{ucs}
              \usepackage[utf8x]{inputenc}
          }
    \else
        \usepackage{fontspec}
        \usepackage{unicode-math}
    \fi

    \usepackage{fancyvrb} 
    \usepackage{grffile} 
    \makeatletter 
    \@ifpackagelater{grffile}{2019/11/01}
    {
    }
    {
      \def\Gread@@xetex#1{%
        \IfFileExists{"\Gin@base".bb}%
        {\Gread@eps{\Gin@base.bb}}%
        {\Gread@@xetex@aux#1}%
      }
    }
    \makeatother
    \usepackage[Export]{adjustbox} 
    \adjustboxset{max size={0.9\linewidth}{0.9\paperheight}}

    \usepackage{hyperref}
    \usepackage{titling}
    \usepackage{longtable} 
    \usepackage{booktabs}  
    \usepackage{array}     
    \usepackage{calc}      
    \usepackage[inline]{enumitem} 
    \usepackage[normalem]{ulem} 
    \usepackage{mathrsfs}
    
      \usepackage{arxiv}

    \definecolor{urlcolor}{rgb}{0,.145,.698}
    \definecolor{linkcolor}{rgb}{.71,0.21,0.01}
    \definecolor{citecolor}{rgb}{.12,.54,.11}

    \definecolor{ansi-black}{HTML}{3E424D}
    \definecolor{ansi-black-intense}{HTML}{282C36}
    \definecolor{ansi-red}{HTML}{E75C58}
    \definecolor{ansi-red-intense}{HTML}{B22B31}
    \definecolor{ansi-green}{HTML}{00A250}
    \definecolor{ansi-green-intense}{HTML}{007427}
    \definecolor{ansi-yellow}{HTML}{DDB62B}
    \definecolor{ansi-yellow-intense}{HTML}{B27D12}
    \definecolor{ansi-blue}{HTML}{208FFB}
    \definecolor{ansi-blue-intense}{HTML}{0065CA}
    \definecolor{ansi-magenta}{HTML}{D160C4}
    \definecolor{ansi-magenta-intense}{HTML}{A03196}
    \definecolor{ansi-cyan}{HTML}{60C6C8}
    \definecolor{ansi-cyan-intense}{HTML}{258F8F}
    \definecolor{ansi-white}{HTML}{C5C1B4}
    \definecolor{ansi-white-intense}{HTML}{A1A6B2}
    \definecolor{ansi-default-inverse-fg}{HTML}{FFFFFF}
    \definecolor{ansi-default-inverse-bg}{HTML}{000000}

    \definecolor{outerrorbackground}{HTML}{FFDFDF}

    \providecommand{\tightlist}{%
      \setlength{\itemsep}{0pt}\setlength{\parskip}{0pt}}
    \DefineVerbatimEnvironment{Highlighting}{Verbatim}{commandchars=\\\{\}}

    \let\Oldtex\TeX
    \let\Oldlatex\LaTeX
    \renewcommand{\TeX}{\textrm{\Oldtex}}
    \renewcommand{\LaTeX}{\textrm{\Oldlatex}}

\makeatletter
\def\PY@reset{\let\PY@it=\relax \let\PY@bf=\relax%
    \let\PY@ul=\relax \let\PY@tc=\relax%
    \let\PY@bc=\relax \let\PY@ff=\relax}
\def\PY@tok#1{\csname PY@tok@#1\endcsname}
\def\PY@toks#1+{\ifx\relax#1\empty\else%
    \PY@tok{#1}\expandafter\PY@toks\fi}
\def\PY@do#1{\PY@bc{\PY@tc{\PY@ul{%
    \PY@it{\PY@bf{\PY@ff{#1}}}}}}}
\def\PY#1#2{\PY@reset\PY@toks#1+\relax+\PY@do{#2}}

\@namedef{PY@tok@w}{\def\PY@tc##1{\textcolor[rgb]{0.73,0.73,0.73}{##1}}}
\@namedef{PY@tok@c}{\let\PY@it=\textit\def\PY@tc##1{\textcolor[rgb]{0.24,0.48,0.48}{##1}}}
\@namedef{PY@tok@cp}{\def\PY@tc##1{\textcolor[rgb]{0.61,0.40,0.00}{##1}}}
\@namedef{PY@tok@k}{\let\PY@bf=\textbf\def\PY@tc##1{\textcolor[rgb]{0.00,0.50,0.00}{##1}}}
\@namedef{PY@tok@kp}{\def\PY@tc##1{\textcolor[rgb]{0.00,0.50,0.00}{##1}}}
\@namedef{PY@tok@kt}{\def\PY@tc##1{\textcolor[rgb]{0.69,0.00,0.25}{##1}}}
\@namedef{PY@tok@o}{\def\PY@tc##1{\textcolor[rgb]{0.40,0.40,0.40}{##1}}}
\@namedef{PY@tok@ow}{\let\PY@bf=\textbf\def\PY@tc##1{\textcolor[rgb]{0.67,0.13,1.00}{##1}}}
\@namedef{PY@tok@nb}{\def\PY@tc##1{\textcolor[rgb]{0.00,0.50,0.00}{##1}}}
\@namedef{PY@tok@nf}{\def\PY@tc##1{\textcolor[rgb]{0.00,0.00,1.00}{##1}}}
\@namedef{PY@tok@nc}{\let\PY@bf=\textbf\def\PY@tc##1{\textcolor[rgb]{0.00,0.00,1.00}{##1}}}
\@namedef{PY@tok@nn}{\let\PY@bf=\textbf\def\PY@tc##1{\textcolor[rgb]{0.00,0.00,1.00}{##1}}}
\@namedef{PY@tok@ne}{\let\PY@bf=\textbf\def\PY@tc##1{\textcolor[rgb]{0.80,0.25,0.22}{##1}}}
\@namedef{PY@tok@nv}{\def\PY@tc##1{\textcolor[rgb]{0.10,0.09,0.49}{##1}}}
\@namedef{PY@tok@no}{\def\PY@tc##1{\textcolor[rgb]{0.53,0.00,0.00}{##1}}}
\@namedef{PY@tok@nl}{\def\PY@tc##1{\textcolor[rgb]{0.46,0.46,0.00}{##1}}}
\@namedef{PY@tok@ni}{\let\PY@bf=\textbf\def\PY@tc##1{\textcolor[rgb]{0.44,0.44,0.44}{##1}}}
\@namedef{PY@tok@na}{\def\PY@tc##1{\textcolor[rgb]{0.41,0.47,0.13}{##1}}}
\@namedef{PY@tok@nt}{\let\PY@bf=\textbf\def\PY@tc##1{\textcolor[rgb]{0.00,0.50,0.00}{##1}}}
\@namedef{PY@tok@nd}{\def\PY@tc##1{\textcolor[rgb]{0.67,0.13,1.00}{##1}}}
\@namedef{PY@tok@s}{\def\PY@tc##1{\textcolor[rgb]{0.73,0.13,0.13}{##1}}}
\@namedef{PY@tok@sd}{\let\PY@it=\textit\def\PY@tc##1{\textcolor[rgb]{0.73,0.13,0.13}{##1}}}
\@namedef{PY@tok@si}{\let\PY@bf=\textbf\def\PY@tc##1{\textcolor[rgb]{0.64,0.35,0.47}{##1}}}
\@namedef{PY@tok@se}{\let\PY@bf=\textbf\def\PY@tc##1{\textcolor[rgb]{0.67,0.36,0.12}{##1}}}
\@namedef{PY@tok@sr}{\def\PY@tc##1{\textcolor[rgb]{0.64,0.35,0.47}{##1}}}
\@namedef{PY@tok@ss}{\def\PY@tc##1{\textcolor[rgb]{0.10,0.09,0.49}{##1}}}
\@namedef{PY@tok@sx}{\def\PY@tc##1{\textcolor[rgb]{0.00,0.50,0.00}{##1}}}
\@namedef{PY@tok@m}{\def\PY@tc##1{\textcolor[rgb]{0.40,0.40,0.40}{##1}}}
\@namedef{PY@tok@gh}{\let\PY@bf=\textbf\def\PY@tc##1{\textcolor[rgb]{0.00,0.00,0.50}{##1}}}
\@namedef{PY@tok@gu}{\let\PY@bf=\textbf\def\PY@tc##1{\textcolor[rgb]{0.50,0.00,0.50}{##1}}}
\@namedef{PY@tok@gd}{\def\PY@tc##1{\textcolor[rgb]{0.63,0.00,0.00}{##1}}}
\@namedef{PY@tok@gi}{\def\PY@tc##1{\textcolor[rgb]{0.00,0.52,0.00}{##1}}}
\@namedef{PY@tok@gr}{\def\PY@tc##1{\textcolor[rgb]{0.89,0.00,0.00}{##1}}}
\@namedef{PY@tok@ge}{\let\PY@it=\textit}
\@namedef{PY@tok@gs}{\let\PY@bf=\textbf}
\@namedef{PY@tok@gp}{\let\PY@bf=\textbf\def\PY@tc##1{\textcolor[rgb]{0.00,0.00,0.50}{##1}}}
\@namedef{PY@tok@go}{\def\PY@tc##1{\textcolor[rgb]{0.44,0.44,0.44}{##1}}}
\@namedef{PY@tok@gt}{\def\PY@tc##1{\textcolor[rgb]{0.00,0.27,0.87}{##1}}}
\@namedef{PY@tok@err}{\def\PY@bc##1{{\setlength{\fboxsep}{\string -\fboxrule}\fcolorbox[rgb]{1.00,0.00,0.00}{1,1,1}{\strut ##1}}}}
\@namedef{PY@tok@kc}{\let\PY@bf=\textbf\def\PY@tc##1{\textcolor[rgb]{0.00,0.50,0.00}{##1}}}
\@namedef{PY@tok@kd}{\let\PY@bf=\textbf\def\PY@tc##1{\textcolor[rgb]{0.00,0.50,0.00}{##1}}}
\@namedef{PY@tok@kn}{\let\PY@bf=\textbf\def\PY@tc##1{\textcolor[rgb]{0.00,0.50,0.00}{##1}}}
\@namedef{PY@tok@kr}{\let\PY@bf=\textbf\def\PY@tc##1{\textcolor[rgb]{0.00,0.50,0.00}{##1}}}
\@namedef{PY@tok@bp}{\def\PY@tc##1{\textcolor[rgb]{0.00,0.50,0.00}{##1}}}
\@namedef{PY@tok@fm}{\def\PY@tc##1{\textcolor[rgb]{0.00,0.00,1.00}{##1}}}
\@namedef{PY@tok@vc}{\def\PY@tc##1{\textcolor[rgb]{0.10,0.09,0.49}{##1}}}
\@namedef{PY@tok@vg}{\def\PY@tc##1{\textcolor[rgb]{0.10,0.09,0.49}{##1}}}
\@namedef{PY@tok@vi}{\def\PY@tc##1{\textcolor[rgb]{0.10,0.09,0.49}{##1}}}
\@namedef{PY@tok@vm}{\def\PY@tc##1{\textcolor[rgb]{0.10,0.09,0.49}{##1}}}
\@namedef{PY@tok@sa}{\def\PY@tc##1{\textcolor[rgb]{0.73,0.13,0.13}{##1}}}
\@namedef{PY@tok@sb}{\def\PY@tc##1{\textcolor[rgb]{0.73,0.13,0.13}{##1}}}
\@namedef{PY@tok@sc}{\def\PY@tc##1{\textcolor[rgb]{0.73,0.13,0.13}{##1}}}
\@namedef{PY@tok@dl}{\def\PY@tc##1{\textcolor[rgb]{0.73,0.13,0.13}{##1}}}
\@namedef{PY@tok@s2}{\def\PY@tc##1{\textcolor[rgb]{0.73,0.13,0.13}{##1}}}
\@namedef{PY@tok@sh}{\def\PY@tc##1{\textcolor[rgb]{0.73,0.13,0.13}{##1}}}
\@namedef{PY@tok@s1}{\def\PY@tc##1{\textcolor[rgb]{0.73,0.13,0.13}{##1}}}
\@namedef{PY@tok@mb}{\def\PY@tc##1{\textcolor[rgb]{0.40,0.40,0.40}{##1}}}
\@namedef{PY@tok@mf}{\def\PY@tc##1{\textcolor[rgb]{0.40,0.40,0.40}{##1}}}
\@namedef{PY@tok@mh}{\def\PY@tc##1{\textcolor[rgb]{0.40,0.40,0.40}{##1}}}
\@namedef{PY@tok@mi}{\def\PY@tc##1{\textcolor[rgb]{0.40,0.40,0.40}{##1}}}
\@namedef{PY@tok@il}{\def\PY@tc##1{\textcolor[rgb]{0.40,0.40,0.40}{##1}}}
\@namedef{PY@tok@mo}{\def\PY@tc##1{\textcolor[rgb]{0.40,0.40,0.40}{##1}}}
\@namedef{PY@tok@ch}{\let\PY@it=\textit\def\PY@tc##1{\textcolor[rgb]{0.24,0.48,0.48}{##1}}}
\@namedef{PY@tok@cm}{\let\PY@it=\textit\def\PY@tc##1{\textcolor[rgb]{0.24,0.48,0.48}{##1}}}
\@namedef{PY@tok@cpf}{\let\PY@it=\textit\def\PY@tc##1{\textcolor[rgb]{0.24,0.48,0.48}{##1}}}
\@namedef{PY@tok@c1}{\let\PY@it=\textit\def\PY@tc##1{\textcolor[rgb]{0.24,0.48,0.48}{##1}}}
\@namedef{PY@tok@cs}{\let\PY@it=\textit\def\PY@tc##1{\textcolor[rgb]{0.24,0.48,0.48}{##1}}}

\makeatother

    \definecolor{incolor}{rgb}{0.0, 0.0, 0.5}
    \definecolor{outcolor}{rgb}{0.545, 0.0, 0.0}

    \hypersetup{
      breaklinks=true,  
      colorlinks=true,
      urlcolor=urlcolor,
      linkcolor=linkcolor,
      citecolor=citecolor,
      }
    
\begin{document}
    
    \maketitle

	\begin{abstract}

We observe that a European Call option with strike $L > K$ can be seen as a Call option with strike $L-K$ on a Call option with strike $K$. Under no arbitrage assumptions, this yields immediately that the prices of the two contracts are the same, in full generality. We study in detail the relative pricing function which gives the price of the Call on Call option as a function of its underlying Call option, and provide quasi-closed formula for those new pricing functions in the Carr-Pelts-Tehranchi family [Carr and Pelts, Duality, Deltas, and Derivatives Pricing, 2015] and [Tehranchi, A Black-Scholes inequality: applications and generalisations, Finance Stoch, 2020] that includes the Black-Scholes model as a particular case. We also study the properties of the function that maps the price normalized by the underlier, viewed as a function of the moneyness, to the normalized relative price, which allows us to produce several new closed formulas. In connection to the symmetry transformation of a smile, we build a lift of the relative pricing function in the case of an underlier that does not vanish. We finally provide some properties of the implied volatility smiles of Calls on Calls and lifted Calls on Calls in the Black-Scholes model.

\end{abstract}

	\hypertarget{introduction}{%
\section{Introduction}\label{introduction}}

	In the early days of equity-to-credit, the \emph{price} of a stock was
modeled as a \emph{Call option} on the value of the underlying company.
This led in turn to the fact that Calls on the stock can be valued with
an option-on-option formula, as obtained by Merton in
\cite{merton1974pricing}. With no relation to this fundamental approach,
we exploit in this note the remark that a (European-type) Call option
with a strike \(L>K\) can be also seen as a Call option with a strike
\(L-K\) on the Call option with strike \(K\) and the same maturity
\(T\); indeed when \(L>K\) it holds for every value of the underlier
\(S(T)\) at maturity:
\[(S(T)-L)_+ = \bigl( (S(T)-K)_+ - (L-K) \bigr)_+.\]

	Under assumptions of perfect market (so that every asset has a single
price, with no bid-asks) and no static arbitrage, this entails that the
\emph{price} (say, at time \(0\)) of the two assets is the same.
Denoting by \(C(S,M)\) the price of standard Calls with strike \(M\)
(the maturity \(T\) is the same for all the contracts), and
\(\hat C_K(\cdot,\star)\) the \emph{relative} pricing function on the
Call on Call contract, it means that the following equality holds:
\[C(S, L) = \hat C_K\bigl(C(S, K), L-K\bigr)\] since the underlier price
of the latter contract is the price of the Call with strike \(K\).

	In \cref{pricing-functions}, we foster this option on option point of
view and obtain, in full generality, relationships between the price of
options on options and the initial Call or Put prices at other strikes:
Calls are Calls on Calls, and Puts are Calls on Puts.

In \cref{the-call-on-call-pricing-function} we define rigorously the
\emph{relative} Call on Call pricing function \(\hat C_K\) and obtain
useful properties in \cref{normalized-call-prices} in the case of
homogeneous pricing function where the option price normalized by the
underlier value does depend only of the moneyness.

In \cref{a-transformation-in-the-tehranchi-space} we show that the
relative Call on Call pricing function leads to a natural transformation
on the space of normalized (by the value of the underlier) Call prices
as functions of the moneyness, that we call the \emph{Tehranchi space},
given by \[\mathbb T_k c (x) := \frac{c\bigl(k+c(k)x\bigr)}{c(k)}\]
where \(k=\frac{K}{S}\) and \(c\) is the normalized Call price. We
provide interesting properties of the transformations \(\mathbb T_k\),
and show that they naturally extend, in some sense, to 2-parameter
transformations. In particular, the derivative at zero with respect to
the moneyness of \(T_k c\) is in general strictly larger than \(-1\),
which corresponds to the fact that the underlier \(C(S,K)\) vanishes
with a positive probability at maturity; the 2-parameter extension
allows to get a derivative at zero equal to \(-1\).

	\Cref{new-closed-formulas} is devoted to the computation of new closed
formulas either for pricing functions or for normalized ones. We provide
a quasi-closed formula when the initial pricing function belongs to the
Carr-Pelts-Tehranchi family, which generalizes the Black-Scholes
formula, obtaining along the way an expression for the underlier value
viewed as a function of the option price for this family.

	In relation to the inversion of the volatility smile in the moneyness
space, there is a generic pricing function transformation which consists
in working in the numeraire of the underlier. We investigate in detail
in \cref{smile-symmetry-and-a-lift-of-the-relative-pricing-function}
this transformation in the case where the underlier may vanish at
maturity, and show that iterating it twice provides a pricing function
on an underlier which does not vanish at maturity. We provide a
quasi-closed formula for the so lifted pricing function in the case of
the Black-Scholes model.

	Eventually we provide basic properties of the volatility smiles
associated to the Black-Scholes relative function and to the lifted
relative one in
\cref{implied-volatility-of-the-relative-pricing-functions}.

	We thank Mehdi El-Amrani for stimulating discussions.

	\hypertarget{pricing-functions}{%
\section{Pricing functions}\label{pricing-functions}}

	We consider general \emph{pricing functions} which give the price
\(C(S,K)\) of a Call option as a function of the underlier price \(S\)
and of its strike \(K\). Of course, the option price may depend on other
variables as well (like the instantaneous variance in stochastic
volatility models as the Heston model), but we will be only interested
in this partial dependency in this case.

	The partial function \(K \to C(S,K)\) gives the Call prices when the
strike varies for the current value of the underlier \(S\), and
typically will aim at \emph{calibrating} the market quotes, whereas the
function \(S \to C(S,K)\) is more interesting in a risk and/or
sensitivity context, e.g.~to get an insight of the order of magnitude of
the tail risk of an option portfolio at horizon one day for margining
purposes.

	\hypertarget{options-are-options-on-options}{%
\subsection{Options are options on
options}\label{options-are-options-on-options}}

	In this section we use the notation \(X(T)\) for the value of the
contract \(X\) at time \(T\). We will drop this notation in the
following sections where \(T\) will not play any role.

	\hypertarget{calls-are-calls-on-calls}{%
\subsubsection{Calls are Calls on
Calls}\label{calls-are-calls-on-calls}}

	Consider \(0<K<L\) and a \emph{European Call on Call contract}, with
strike \(L-K\), which delivers at maturity \(T\) a Call contract with
strike \(K\). The payoff at \(T\) of this contract will be
\((C(S,K)(T)-(L-K))_+\).

Observe now that \(C(S,K)(T) = (S(T)-K)_+\) so that the payoff of the
contract is equal to \(((S(T)-K)_+-(L-K))_+\). Now this latter quantity
is \(0\) if \(S(T) \leq L\) and \(S(T)-L\) otherwise, so it is equal to
\((S(T)-L)_+=C(S,L)(T)\), i.e.
\[\forall 0<K<L,\; \; C(S,L)(T) = \bigl(C(S,K)(T)-(L-K)\bigr)_+.\]

So in full generality under perfect market assumptions for the underlier
and the options \(C(S,K)\) and \(C(S,L)\), \emph{the price of a Call
option with strike \(L-K\) on a Call option with strike \(K\) is the
price of a Call option with strike \(L\)}.

	What happens if \(L\) is smaller than \(K\)? There is no hope to get any
relation in this case, since the option price \(C(S,K)(T)\) will vanish
in the range \([L,K]\) where the Call \(C(S,L)(T)\) will not. So, in
terms of smiles, for a fixed value of \(K\), only the part of the smile
\emph{on the right of \(K\)} will give rise to a new smile.

	What happens for Put prices?

	\hypertarget{put-call-parity-and-the-put-price}{%
\subsubsection{Put-Call-Parity and the Put
price}\label{put-call-parity-and-the-put-price}}

	Let us denote \(\hat C_K(X,M)\) and \(\hat P_K(X,M)\) the Call and Put
pricing functions for Calls and Puts with a strike \(M\) on a Call
option \(X=C(S,K)\) with strike \(K\). The Put-Call-Parity reads
\[\hat C_K(X,L-K)-\hat P_K(X,L-K) = C(S,K)-(L-K).\] Now
\(\hat C_K(X,L-K)=C(S,L)\) and using the classic Put-Call-Parity at
strikes \(K\) and \(L\) yields, taking the difference:
\(C(S,L)-C(S,K)=P(S,L)-P(S,K) - (L-K)\). This implies that
\[\hat P_K(X,L-K) = P(S,L)-P(S,K).\]

	This relation clarifies what the price of the Put is in the new world
where the underlier is the option with strike \(K\), but also provides
insights on the properties of the difference \(P(S,L)-P(S,K)\). Can we
prove it directly? Yes, indeed if we look at the difference
\((L-S(T))_+-(K-S(T))_+\), it is constantly equal to \(L-K\) below
\(K\), and then goes to \(0\) linearly at point \(L\), where it remains.
This can be viewed also as a function of \((S(T)-K)_+\), which is
exactly a Put payoff with strike \(L-K\). In other words, it holds that
\[(L-S(T))_+-(K-S(T))_+ =\bigl((L-K) - (S(T)-K)_+\bigr)_+,\] which gives
another proof of the relation \(\hat P_K(X,L-K) = P(S,L)-P(S,K)\).

	Eventually, summarizing the Call an Put computations we have the
property that \begin{align*}
\hat C_K(X,L-K) &= C(S,L) \\
\hat P_K(X,L-K) &= P(S,L)-P(S,K)
\end{align*} or yet for any \(M \geq 0\) \begin{align*}
\hat C_K(X,M) &= C(S,K+M) \\
\hat P_K(X,M) &= P(S,K+M)-P(S,K).
\end{align*}

	\hypertarget{calls-on-calls-further-iterations}{%
\paragraph{Calls on Calls: further
iterations}\label{calls-on-calls-further-iterations}}

	Considering now a Call option with strike \(N\) written on a Call
\(\hat C_K(X,M)\), from the above equation this is equivalent to a Call
of the form \(\hat C_{K+M}(X,N)\). The latter quantity again is
equivalent to \(C(S,K+M+N)\). Similarly, a Put option with strike \(N\)
written on a Call \(\hat C_K(X,M)\) is a Put option written on
\(C(S,K+M)\), so it equals \(\hat P_{K+M}(X,N)\), or
\(P(S,K+M+N)-P(S,K+M)\).

We have therefore a semigroup property, and iterating further does not
yield new pricing functions.

	\hypertarget{puts-are-calls-on-puts}{%
\subsubsection{Puts are Calls on Puts}\label{puts-are-calls-on-puts}}

	Consider now \(Y=P(S,K)\) as an underlier. Can we mimic the above
approach using Puts as underliers? Observe first that Put prices are
bounded by the strike, so that we have an underlier with values in
\([0,K]\).

Take now any strike \(0 \leq L<K\). Then\\
\[\bigl((K-S(T))_+-L\bigr)_+ = \bigl((K-L)-S(T)\bigr)_+\] which gives
that a Call on \(P(S,K)\) with strike \(L\) is a Put on \(S\) with
strike \(K-L\).

	This entails, if we denote by \(\tilde C_K(Y,L)\) the price of this
Call, that \(\tilde C_K(Y,L)=P(S,K-L)\). To get the price
\(\tilde P_K(Y,L)\) of the corresponding Put, let us use the
Put-Call-Parity as above:
\(\tilde C_K(Y,L)-\tilde P_K(Y,L) = P(S,K) - L\). Using the classic
Put-Call-Parity at the strikes \(K-L\) and \(K\), we find
\[\tilde P_K(Y,L) =  P(S,K-L)-P(S,K) + L=C(S,K-L)-C(S,K).\]

	We get eventually another pair transform:

\(\forall 0 \leq L<K\), \begin{align*}
\tilde C_K(Y,L) &= P(S,K-L) \\
\tilde P_K(Y,L) &= C(S,K-L)-C(S,K).
\end{align*}

	\hypertarget{puts-on-puts-further-iterations}{%
\paragraph{Puts on Puts: further
iterations}\label{puts-on-puts-further-iterations}}

	Consider now a Call option with strike \(N<L<K\) written on
\(\tilde P_K(Y,L)\). Looking at the payoff function, it can be easily
shown that
\[\bigl(\bigl(L-(K-S(T))_+\bigr)_+-N\bigr)_+ = \bigl(S(T)-(K-(L-N))\bigr)_+ - \bigl(S(T)-K\bigr)_+.\]
In other words, such a Call has the same value as the portfolio
\(C(S,K-(L-N))-C(S,K)\), which in turn we have shown to be equal to
\(\tilde P_K(Y,L-N)\). Again, using the Put-Call-Parity, we find that a
Put with strike \(N\) on \(\tilde P_K(Y,L)\) is equivalent to the
portfolio \(P(S,K-(L-N))-P(S,K-L)\), i.e.~to
\(\tilde C_K(Y,L-N) - \tilde C_K(Y,L)\).

Summarizing, we get the following relationships: \begin{align*}
\text{Call on $\tilde P_K(Y,L)$}(N) &= C(S,K-(L-N))-C(S,K)\\
\text{Put on $\tilde P_K(Y,L)$}(N) &= P(S,K-(L-N))-P(S,K-L).
\end{align*}

	\hypertarget{the-call-on-call-pricing-function}{%
\subsection{The Call on Call pricing
function}\label{the-call-on-call-pricing-function}}

	In order to define rigorously the \emph{relative} function \(\hat C_K\)
of Calls on Calls, we need to assume for a while that the function
\(S \to C(S, K)\) is invertible (it is the case in the Black-Scholes
model and many other ones).

\begin{definition}

Let $X = C(S,K)$ and $M>0$. We denote with $\hat C_K(X,M)$ the Call option with strike $M$ on the Call option with strike $K$. Then $\hat C_K(X,M)$ is the price of a Call option with strike $K+M$ and underlier $X$. In particular, if the function $S \to C(S, K)$ is invertible it holds
\begin{equation}\label{eqHatCnoS}
\hat C_K(X, M) := C\bigl(C^{-1}(X,K),K+M\bigr).
\end{equation}
$K$ is called the relative underlying strike of $\hat C_K(X, M)$ and $X$ the underlier of $\hat C_K(X, M)$.

\end{definition}

\Cref{eqHatCnoS} gives a first representation of \(\hat C_K\). It is of
little practical interest though, since we are not aware of any model
where both \(C^{-1}\) and \(C\) can be computed explicitly.
Nevertheless, we will see in \cref{the-carr-pelts-tehranchi-family} that
in the vast class of pricing functions of the Carr-Pelts-Tehranchi
family a convenient representation formula for the inverse function is
available.

We investigate below general properties of the Call on Call pricing
function.

	\hypertarget{properties-of-the-call-on-call-pricing-function}{%
\subsubsection{Properties of the Call on Call pricing
function}\label{properties-of-the-call-on-call-pricing-function}}

	From the arguments of \cref{calls-are-calls-on-calls}, we can deduct
some first properties of the function \(\hat C_K\). In particular, Calls
on Calls satisfy the usual arbitrage bounds for Call prices, i.e.~they
are always larger than their intrinsic value and smaller than the
underlier. Furthermore, they are convex and non-increasing as function
of the strike. We already expect these properties to hold true for
arbitrage arguments, and we show them rigorously in the following
proposition.

	\begin{proposition}[Relative pricing function: strike dependence]\label{propEasyCallsOnCalls}

The function $M \to \hat C_K(X, M)$ satisfies
$$(X-M)_+ \leq \hat C_K(X, M) \leq X$$
and it is convex, non-increasing, with a slope strictly larger than $-1$.

\end{proposition}

	\begin{proof}

The function $M\to C(S,K+M)$ is convex and non-increasing, and so is $M \to \hat C_K(X, M)$. This can be proved also observing that the basic relations $(S-K)_+ \leq C(S,K) \leq S$ translate into $(C(S,K) - M)_+ \leq C(S,K+M) \leq C(S,K)$, which gives in particular that the function $M \to \hat C_K(X, M)$ is non-increasing. Furthermore, from \cref{eqHatCnoS}, $\frac{d}{dM}\hat C_K(X, M) = \partial_K C(C^{-1}(X,M),K+M) \geq \partial_K C(C^{-1}(X,M),K)$ which is strictly larger than $-1$. 

\end{proof}

	Observe that the above inequality implies
\((S-(K+M))_+ \leq \hat C_K(X, M) \leq S\) since
\((S-(K+M))_+\leq(C(S,K)-M)_+\) and \(S\geq C(S,K)=X\).

	\begin{remark}\label{remarkDeriv1}

\Cref{propEasyCallsOnCalls} implies in particular that the slope of Calls on Calls in $0$ is stricly larger than $-1$. This is not a problem in terms of arbitrageable prices, but it is an uncommon feature since it is linked to the presence of a positive mass of the underlier in $0$ (see Theorem 2.1.2. of \cite{tehranchi2020black}). This is expected indeed, since the new underlier is a Call option, which has a whole region of null payoff. 
We will target this point in \cref{smile-symmetry-and-a-lift-of-the-relative-pricing-function} where we will define lifted Calls on Calls' prices with derivative equal to $-1$ at $0$.

\end{remark}

	In the following we identify the necessary and sufficient conditions
that the Call on Call pricing function must satisfy in order to be
monotone as a function of the relative underlying strike and convex as a
function of the underlier, i.e.~the original Call price.

	\begin{lemma}[Monotonicity with respect to the relative underlying strike]\label{propMonotonicityStrike}

Assuming the $\mathcal C^1$ smoothness of $C(\cdot, K)$ and $C(S, \cdot)$, the function $K\to\hat C_K(X,M)$ is non-decreasing if and only if the function
$$L\to\frac{\partial_KC(S,L)}{\partial_{S}C(S,L)}$$
is non-decreasing for every $S$.

\end{lemma}

	\begin{proof}

Firstly observe that $C(C^{-1}(X,K),K)=X$, so that taking the derivative with respect to $K$ we find
$$0 = \partial_{S}C(C^{-1}(X,K),K)\partial_KC^{-1}(X,K) + \partial_KC(C^{-1}(X,K),K)$$
or
$$\partial_KC^{-1}(X,K) = -\frac{\partial_KC(C^{-1}(X,K),K)}{\partial_{S}C(C^{-1}(X,K),K)}.$$
We can now consider the relation $\hat C_K(X, M) = C(C^{-1}(X,K), K+M)$ and develop the derivative with respect to $K$:
$$\frac{d}{dK}\hat C_K(X, M) = -\frac{\partial_{S}C(C^{-1}(X,K),K+M)\partial_KC(C^{-1}(X,K),K)}{\partial_{S}C(C^{-1}(X,K),K)} + \partial_KC(C^{-1}(X,K),K+M).$$
Then, $\hat C_K(X, M)$ is non-decreasing as a function of $K$ iff
$$\frac{\partial_KC(C^{-1}(X,K),K)}{\partial_{S}C(C^{-1}(X,K),K)}\leq\frac{\partial_KC(C^{-1}(X,K),K+M)}{\partial_{S}C(C^{-1}(X,K),K+M)},$$
or equivalently iff the function
$$L\to\frac{\partial_KC(S,L)}{\partial_{S}C(S,L)}$$
is non-decreasing for every $S$.

\end{proof}

	\begin{lemma}[Convexity with respect to the underlier]\label{propConvexX}

Assuming the $\mathcal C^2$ smoothness of $C(\cdot , K)$, the function $X\to\hat C_K(X,M)$ is convex if and only if the function
$$K \to \frac{\partial_{S}^2 C(S, K)}{\partial_{S} C(S, K)}$$
is non-decreasing for every $S$.

\end{lemma}

	\begin{proof}

Let us restart from $\hat C_K(X, M) = C(C^{-1}(X,K), K+M)$. Assuming the $\mathcal C^2$ smoothness of $C(\cdot , K)$ we get
$$\frac{d}{dX} \hat C_K(X, M) = \frac{\partial_{S} C(C^{-1}(X,K), K+M)}{\partial_{S} C(C^{-1}(X,K), K)},$$
so that $\frac{d^2}{dX} \hat C(X, M)$ has the sign of the quantity
$$\partial_{S}^2 C(C^{-1}(X,K), K+M) \partial_{S} C(C^{-1}(X,K), K) - \partial_{S} C(C^{-1}(X,K), K+M) \partial_{S}^2 C(C^{-1}(X,K), K).$$
As a consequence, the function $X \to \hat C_K(X, M)$ is convex for any $K,M$ iff the function 
$$L \to \frac{\partial_{S}^2 C(C^{-1}(X,K), L)}{\partial_{S} C(C^{-1}(X,K), L)}$$
is non-decreasing for any $X$, which is equivalent to state the same property for the function
$$K \to \frac{\partial_{S}^2 C(S, K)}{\partial_{S} C(S, K)}$$
at any point $S$.

\end{proof}

	In the next section we will apply
\cref{propMonotonicityStrike,propConvexX} to the case of homogeneous
pricing functions, and in particular to the Black-Scholes case for which
the properties of monotonicity with respect to the relative underlying
strike and of convexity with respect to the underlier are always
satisfied.

	\hypertarget{normalized-call-prices}{%
\subsection{Normalized Call prices}\label{normalized-call-prices}}

	We now switch from the strike space to the moneyness \(k=\frac KS\)
space and consider \emph{normalized} Call pricing functions, i.e.~Call
prices divided by their underlier.

\begin{definition}\label{defNormalization}

Let $k=\frac{K}{S}$ the moneyness of the Call option $C(S,K)$, and $m=\frac{M}{C(S,K)}$ the moneyness of the Call option $\hat C_K(C(S,K),M)$. We denote with
$$C_{S}(k):=\frac{C(S,Sk)}{S}$$
the normalization of $C$ with respect to $k$, and with
$$\hat C_{K,S}(m) := \frac{\hat C_K(C(S,K),C(S,K)m)}{C(S,K)}$$
the normalization of $\hat C_K$ with respect to $m$.

Furthermore, we say that $C$ is homogeneous if $C_S$ does not depend on $S$ and define the normalized pricing function $c$ by the relation $c(k):=C_1(k)$ for every $k$.

\end{definition}

	Normalized Call prices are particularly interesting when Call prices are
homogeneous, since they satisfy key properties as we will show in
\cref{homogeneous-call-prices}. Furthermore, the most notorious models
such as the Black-Scholes, the Heston and the implied volatility models
are homogeneous. In this case the function \(C\) can be recovered from
\(c\) through the formula \(C(S,K) = Sc\bigl(\frac{K}{S}\bigr)\). Not
all models are homogeneous: examples of inhomogeneous models include
local volatility or local stochastic volatility (except in very few
cases).

In order to work with Black-Scholes prices, throughout the rest of the
paper we denote with \(\phi\) the standard normal probability density
function and with \(\Phi\) its cumulative density function. Furthermore,
we denote with \(\text{BS}(S,K,v)\) the traditional Black-Scholes
function for Call prices with implied total volatility
\(v=\sigma\sqrt T\): \begin{equation}\label{eqBS}
\begin{aligned}
\text{BS}(S,K,v) &= S\Phi\bigl(d_1(S,K,v)\bigr) - K\Phi\bigl(d_2(S,K,v)\bigr)\\
d_{1,2}(S,K,v) &= -\frac{\log\frac{K}{S}}{v}\pm\frac{v}{2}.
\end{aligned}
\end{equation} We will sometimes drop the dependency in \(v\) for
notation simplicity. When considering normalized Black-Scholes prices,
we use the notation
\[\text{BS}(S,K) = S \text{bs}\biggl(\frac{K}{S}\biggr).\] Reconstructed
prices obtained from \(\text{bs}(k)\) in the Black-Scholes case
correspond to the perspective function of section 3.2.6. of
\cite{boyd2004convex}.

	In the following lemma we show that the normalization \(C_S(k)\) as a
function of the moneyness \(k\) has the same properties as the original
price \(C(S,K)\) as a function of the strike \(K\).

\begin{lemma}\label{lemmaPropertiesNormalized}

Normalized prices $C_S(k)$ are non-increasing and convex functions of $k$, and satisfy
$$(1-k)_+ \leq C_S(k) \leq 1.$$

\end{lemma}

\begin{proof}

It holds $C_S'(k) = \frac{d}{dK}C(S,Sk)$ and $C_S''(k) = S\frac{d^2}{dK^2}C(S,Sk)$, so that $C_S(k)$ is non-increasing and convex in $k$. Also, since Call prices satisfy $(S-K)_+ \leq C(S,K) \leq S$, then dividing by $S$ it holds $(1-k)_+ \leq C_S(k) \leq 1$.

\end{proof}

	It turns out that there is a convenient relationship between the initial
normalized Call pricing function and the (normalized) relative Call on
Call one. Indeed, observe that \(\hat C_K(C(S,K),M) = C(S, K+M)\), so
that
\(\hat C_K(C(S,K),C(S,K)m)= C\bigl(S, S \frac{K+C(S,K)m}{S}\bigr)\). Now
from \cref{defNormalization} it holds
\[C_{S}(k+C_{S}(k) m) = \frac{C \bigl(S,S (k+C_{S}(k) m) \bigr)}{S}\]
and consequently
\[\hat C_{Sk,S}(m) = \frac{C_{S}(k + C_{S}(k) m)}{C_{S}(k)}.\] In
particular, in case \(C\) is homogeneous,
\begin{equation}\label{eqCHatnorm}
\hat C_{Sk,S}(m)= \frac{c(k+c(k) m)}{c(k)}.
\end{equation}

	We will further exploit the relationship in \cref{eqCHatnorm} in
\cref{a-transformation-in-the-tehranchi-space} where we work in the
space of normalized homogeneous Call prices, and define transformations
in such space.

As in \cref{remarkDeriv1}, observe that by the chain rule it holds that
\(\frac{d}{dm} \hat C_{Sk,S}(0_+) = C'_S(k) = \partial_K C(S, Sk)\),
which will be in general (for strictly convex functions) strictly larger
than \(-1\).

	\hypertarget{homogeneous-call-prices}{%
\subsubsection{Homogeneous Call prices}\label{homogeneous-call-prices}}

	We now look at the properties of the Call on Call pricing functions in
\cref{properties-of-the-call-on-call-pricing-function} in the case of
homogeneous Call prices. It turns out that Calls on Calls with
homogeneous pricing function are non-decreasing functions of the
relative underlying strike and that Black-Scholes Calls on Calls are
convex with respect to the underlier. As a consequence, when calibrating
Calls on Calls, one should design an algorithm such to satisfy these
necessary properties.

	In the following proposition we consider conditions of
\cref{propMonotonicityStrike} in the case of homogeneous Call prices and
show that they are always satisfied, i.e.~that Calls on Calls with
homogeneous pricing function are non-decresing with respect to the
relative underlying strike.

	\begin{proposition}[Monotonicity of Calls on Calls with respect to the relative underlying strike: the homogeneous case]\label{corolMonotonicity}

Let $C(S,K)$ homogeneous and $\mathcal C^1$ in both variables. Then the function $K\to\hat C_K(X,M)$ is non-decreasing.

\end{proposition}

	\begin{proof}

From \cref{propMonotonicityStrike}, we shall prove that the function
$$L\to\frac{\partial_KC(S,L)}{\partial_{S}C(S,L)}$$
is non-decreasing for every $S$. In the homogeneous case this can be simplified writing $C(S,L)=c\bigl(\frac{L}{S}\bigr)S$ and considering that a function is monotone in $L$ iff it is monotone in $l=\frac{L}{S}$. We then find that the function $K\to \hat C_K(X,M)$ is non-decreasing iff the function
$$l\to \frac{c'(l)}{c(l)-lc'(l)}$$
is non-decreasing. This is actually the case since the derivative of the latter function is $\frac{c(l)c''(l)}{(c(l)-lc'(l))^2}$, which is always positive for convex prices.

\end{proof}

	Consider now a fixed value of \(X\). This property gives that at a fixed
moneyness \(\frac{M}{X}\), the map \(K \to \hat C_K(X,M)\) is
non-decreasing and so, for any continuous increasing function \(Y\) with
\(Y(0)=0\), the map \(t \to \hat C_{Y(t)}(X,M)\) is non-decreasing as
well, meaning there is no calendar-spread arbitrage for the price
surface \((t,M) \to C_{Y(t)}(X,M)\). Since there is no Butterfly
arbitrage in the strike dimension for any \(t\), we have built an
arbitrage-free \emph{forward extrapolation} of the pricing function
\(C_0(X,M) = C(X,M)\). One can see that we treat the strike \(K\) here
as a shadow parameter, completely forgetting its role in the design of
the relative pricing function.

	We now pass to the study of \cref{propConvexX}. Conditions for the
convexity of Calls on Calls with respect to the underlier can be
re-written in the homogeneous case. Differently from the property of
monotonicity with respect to the relative underlying strike, here we do
not achieve to show the convexity property for all homogeneous pricing
function. However, we prove it for the Black-Scholes case.

	\begin{proposition}[Convexity of Calls on Calls with respect to the underlier: the homogeneous case]

Let $C(S,K)$ homogeneous and $\mathcal C^2$ in the first variable. Then the function $X\to\hat C_K(X,M)$ is convex if and only if the function
$$k \to \frac{k^2c''(k)}{c(k)-kc'(k)}$$
is non-decreasing. In particular this holds true in the Black-Scholes case.

\end{proposition}

	\begin{proof}

From \cref{propConvexX} we shall prove that the function
$$K \to \frac{\partial_{S}^2 C(S, K)}{\partial_{S} C(S, K)}$$
is non-decreasing for every $S$. We can write $C(S,K)=c\bigl(\frac{K}{S}\bigr)S$, develop the derivatives and consider that a function is monotone in $K$ iff it is monotone in $\frac{K}{S}$. We find that in the homogeneous case, $X \to \hat C_K(X, M)$ is convex for any $K,M$ iff the function
$$k \to \frac{k^2c''(k)}{S(c(k)-kc'(k))}$$
is non-decreasing for any $S$. We can drop $S$ at the denominator and conclude.

In the Black-Scholes case, $\text{bs}''(k)=\frac{\phi(d_2)}{kv}$ where $d_{1,2}=-\frac{\log k}{v} \pm \frac{v}2$. Then the above requirement is that
$$k\to\frac{k\phi(d_2)}{v\Phi(d_1)}$$
is non-decreasing. This holds true iff, taking the derivative, the quantity
$$\frac{\phi(d_2)}{v\Phi(d_1)^2}\biggl(\Phi(d_1)+\frac{\Phi(d_1)d_2+\phi(d_1)}{v}\biggr)$$
is positive. Observe that $d_2=d_1-v$, so that we are asking the quantity $\Phi(d_1)d_1+\phi(d_1)$ to be positive. When $d_1$ is positive this is gained. Otherwise, we can use the upper bound of the Mill's ratio $\frac{1-\Phi(x)}{\phi(x)}<\frac{1}{x}$ for every $x>0$ with $x=-d_1$ and obtain the desired property.

\end{proof}

	The convexity in the underlier of the option price is a key property
from a risk analysis perspective, and allows to study the behavior of
the option price dynamic as being locally Black-Scholes-like, with a
positive Gamma for Calls and Puts. Combined with the previous
proposition and the discussion that follows it, we get a forward
extrapolation scheme with nice properties when the convexity property is
fulfilled.

	\begin{remark}

It is interesting to observe that the function $\hat C_K(X,M)$ cannot be homogeneous when $C(S, K)$ is. Indeed, in order to satisfy such a property, its normalized function $\hat C_{K,S}(m) = \frac{\hat C_K(X,Xm)}{X}$ where $S$ is recovered from $X=C(S,K)$ should not depend on $X$, i.e. it should be a function of the form $g(m)$. From \cref{eqCHatnorm}, it should hold $g(m)=\frac{c\bigl(\frac{K}{S} + c\bigl(\frac{K}{S}\bigr)m\bigr)}{c\bigl(\frac{K}{S}\bigr)}$. However the right term depends on $X$ in the $S$ term, so that the equality cannot hold for all $X$. Indeed, for $X$ moving from its lowest values to $\infty$, $S$ moves from $0$ to $\infty$, so that $g(m)=\frac{c(\infty)}{c(\infty)}=1$ and $g(m)=c(m)$ respectively. In non-degenerate cases, normalized prices are not constantly equal to $1$ so that Calls on Calls with homogeneous pricing function cannot be homogeneous.

\end{remark}

	\hypertarget{a-transformation-in-the-tehranchi-space}{%
\subsection{A transformation in the Tehranchi
space}\label{a-transformation-in-the-tehranchi-space}}

	In \cref{lemmaPropertiesNormalized} we have pointed out some necessary
properties that normalized Call prices satisfy: monotonicity and
convexity with respect to the moneyness, and upper and lower bounds
corresponding to the constant function \(\mathbbm 1\) and the normalized
intrinsic value function \((1-k)_+\). Note that the property of
monotonicity is actually implied by the two other properties.

A crucial point here is that the underlier is considered to be frozen
(and, given the normalization, with unit value): in other words we only
consider the partial dependency in the normalized strike (the moneyness)
of the pricing function.

As Tehranchi has deeply studied normalized Call prices in
\cite{tehranchi2020black}, we will name \emph{Tehranchi space} the space
\(\mathbb C\) of such normalized Call prices:
\[\mathbb C = \Bigl\{c:\mathbb R_+\longrightarrow[0,1]\ \bigl|\,  c \; \text{convex},\; \forall m, \; (1-m)_+\leq c(m)\leq 1 \Bigr\}.\]

As an immediate consequence, functions in \(\mathbb C\) are
non-increasing and satisfy \(c(0)=1\). Also, from
\cref{lemmaPropertiesNormalized}, functions obtained by the
normalization \(c(k)\) of homogeneous prices defined in
\cref{defNormalization} belong to the Tehranchi space.

	\Cref{eqCHatnorm} suggests to define the following transformation on
\(\mathbb C\).

\begin{definition}

For any $c \in \mathbb C$ and $k\geq0$ with $c(k)>0$ we define the transformation
$$\mathbb T_k c(\cdot)  := \frac{c(k + c(k) \cdot )}{c(k)}.$$
$k$ is called the relative underlying moneyness of $\mathbb T_k$.

\end{definition}

	Observe that functions in \(\mathbb C\) are either positive, or positive
before a threshold \(a\) and null beyond \(a\). It is natural if needed
to extend the definition of \(\mathbb T_k\) for \(k \geq a\) by
\(\mathbb T_k \equiv \mathbbm{1}\), the constant function equal to the
normalized underlier.

In relation to \cref{eqCHatnorm}, the transformation \(\mathbb T_k\)
corresponds to the normalization of Calls on Calls with homogeneous
pricing function, i.e.~\(\mathbb T_k c(m) = \hat C_{Sk,S}(m)\). Also,
for a given \(S\) and a function \(\mathbb T_kc(\cdot)\), it is always
possible to reconstruct the corresponding non-normalized Call on Call.
In particular, the original underlier Call written on \(S\) has strike
\(K=Sk\), and the Call on Call with strike \(M\) is
\[\hat C_K(C(S,K),M) = \hat C_K(Sc(k),M) = \mathbb T_kc\Bigl(\frac{M}{Sc(k)}\Bigr)Sc(k).\]

	\hypertarget{properties-of-the-transformation-mathbb-t_k}{%
\subsubsection{\texorpdfstring{Properties of the transformation
\(\mathbb T_k\)}{Properties of the transformation \textbackslash mathbb T\_k}}\label{properties-of-the-transformation-mathbb-t_k}}

	The following lemma lists important properties of the transformations
\(\mathbb T_k\). In particular, it states that the new function
\(\mathbb T_k c\) still lives in \(\mathbb C\) and that it has a
derivative in \(0\) which is larger than \(-1\). Furthermore the lemma
gives the limits of the transformation \(\mathbb T_k\) with respect to
\(k\).

	\begin{lemma}[Properties of $\mathbb T_k$]\label{lemmaT}

For any $c \in \mathbb C$ and $k\geq0$ with $c(k)>0$ it holds:
\begin{enumerate}
\item $\mathbb T_k c \in \mathbb C$;
\item $\frac{d}{dm} \mathbb T_kc(0_+) = c'(k)$;
\item $\mathbb T_kc(\infty) = \frac{c(\infty)}{c(k)}$;
\item $\mathbb T_0 c=c$;
\item $\mathbb T_\infty c \equiv \mathbbm{1}$.
\end{enumerate}

\end{lemma}

	\begin{proof}

The only difficult point is the last one. Observe that at fixed $k$, $m$, it holds $c(k+c(k)m)=c(k)+c'(k+uc(k))c(k)m$ for some $u$ in $]0,m[$. Whence $\mathbb T_k c(m)=1+c'(k+uc(k))m$ and since $c'$ goes uniformly to $0$ at infinity this yields $\mathbb T_k c(\cdot) \to 1$ as $k \to \infty$.

\end{proof}

	At this point, one can consider the family
\(\{\mathbb T_k c: k \geq 0 \}\), as a (one-dimensional)
\emph{enrichment} of the price curve \(c\), given that
\(\mathbb T_0 c=c\). The initial forward moneyness \(k\) should be
considered here as a plain parameter; all the price curves
\(\mathbb T_k c\) are arbitrage-free in the sense that they belong to
\(\mathbb C\).

	In relation to this latter point, one can wonder about the
\emph{composition} of the above enrichment/extensions, like
\(\mathbb T_{k_n} \cdots \mathbb T_{k_2} \mathbb T_{k_1}\). The
following property corresponds to the image of the semigroup property in
the normalized space:

\begin{lemma}[Iterates of $\mathbb T_k$]\label{lemmaTSemigroup}

It holds $\mathbb T_b \mathbb T_a c = \mathbb T_{a+c(a)b}c$.

\end{lemma}

\begin{proof}

The following relations hold
\begin{align*}
\mathbb T_b\mathbb T_a c(m) &= \frac{\mathbb T_a c(b+\mathbb T_a c(b)m)}{\mathbb T_a c(b)}\\
&= \mathbb T_a c\Bigl(b+\frac{c(a+c(a)b)}{c(a)}m\Bigr)\frac{c(a)}{c(a+c(a)b)}\\
&= \frac{c\Bigl(a+c(a)\bigl(b+\frac{c(a+c(a)b)}{c(a)}m\bigr)\Bigr)}{c(a)}\frac{c(a)}{c(a+c(a)b)}\\
&= \frac{c(a+c(a)b+ c(a+c(a)b)m}{c(a+c(a)b)}\\
&= \mathbb T_{a+c(a)b}c(m).
\end{align*}

\end{proof}

	This means that the range of \(\mathbb T_.\) is the same as the range of
its iterates, and there is no \emph{additional} enrichment to hope for
from performing those iterations.

	We shall now consider \cref{corolMonotonicity} where we proved that
Calls on Calls with homogeneous pricing function are non-decreasing with
respect to the relative underlying strike. We expect to find a similar
property for \(k\to\mathbb T_k c(m) = \hat C_{Sk,S}(m)\).

	\begin{proposition}[Monotonicity of $\mathbb T_k$ with respect to the relative underlying moneyness]\label{propTincreasing}

For any $c\in\mathbb C$ and $m\geq 0$, the map $k\to\mathbb T_kc(m)$ is non-decreasing.

\end{proposition}

	\begin{proof}

It holds
$$\frac{d}{dk}\mathbb T_kc(m) = \frac{c'(k+c(k)m)(1+c'(k)m)c(k) - c(k+c(k)m)c'(k)}{c(k)^2}.$$
Doing the derivative with respect to $m$, one finds $\frac{d}{dm}\frac{d}{dk}\mathbb T_kc(m)=c''(k+c(k)m)(1+c'(k)m)$ which is positive iff $m<-\frac{1}{c'(k)}$. Then, the function $\frac{d}{dk}\mathbb T_kc(m)$ with variable $m$ is increasing up to $-\frac{1}{c'(k)}$ and then starts decreasing. To show that it is non-negative for every $k$ and $m$, it is enough to show that it is non-negative for every $k$ and $m\in\{0,\infty\}$.

At $m=0$, it is easy to see $\frac{d}{dk}\mathbb T_kc(0) = 0$. From Theorem 2.1.2 of \cite{tehranchi2020black}, there exists a random variable $S$ such that $c(y) = 1-E[S\land y]$ and $-c'(y) = P(S>y)$. Then
\begin{align*}
c(y) - yc'(y) &= 1- E[S\land y]+yP(S>y)\\
&= 1- \int(s\land y)f_S(s)\,ds + y\int_y^\infty f_S(s)\,ds\\
&= 1-\int_0^ysf_S(s)\,ds-y\int_y^\infty f_S(s)\,ds + y\int_y^\infty f_S(s)\,ds\\
&= 1 -\int_0^ysf_S(s)\,ds
\end{align*}
and this goes to $1-E[S] = c(\infty)\geq 1$ as $y$ goes to $\infty$. As a consequence, $yc'(y)$ goes to $0$ and the result holds.

\end{proof}

	Note that the above proposition implies in particular
\[c(m) = \mathbb T_0c(m) \leq \mathbb T_kc(m) = \frac{c(k+c(k)m)}{c(k)}.\]

	In \cref{figureT} we plot the function \(\mathbb T_kc(m)\) with respect
to \(k\) for different fixed \(m\)s. The function \(c\) is a normalized
Black-Scholes Call function with implied total volatility equal to
\(0.2\). It can be seen that \(\mathbb T_kc(m)\) is non-decreasing in
\(k\) (as shown in \cref{propTincreasing}) and non-increasing in \(m\),
as expected since \(c\) is a non-increasing function.

\begin{figure}
	\centering
	\includegraphics[width=.5\linewidth]{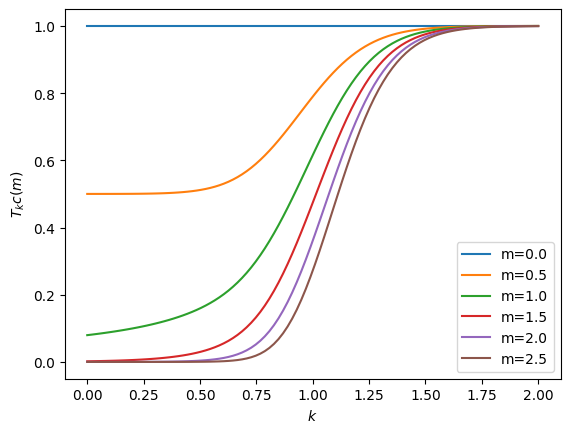}
	\caption{Function $k\to \mathbb T_kc(m)$ for different levels of $m$. The original Black-Scholes implied total volatility is set at $0.2$.}
	\label{figureT}
\end{figure}
    
	\hypertarget{a-slight-generalization}{%
\subsubsection{A slight generalization}\label{a-slight-generalization}}

	From \cref{lemmaT}, the function \(m\to\mathbb T_kc(m)\) has a
particular feature at \(0\). Indeed, its right derivative is
\(\frac{d}{dm} \mathbb T_kc(0_+) = c'(k)\) which for \(k>0\) is in
general larger than \(-1\). As already seen in \cref{remarkDeriv1}, this
feature might be annoying since it implies the presence of a mass in
\(0\) of the probability density function associated to the underlier of
prices.

We are then interested in generalizing suitably the transformation
\(\mathbb T_k\) in order to get rid of this mass at \(0\) phenomenon.
This generalization is formulated on the Tehranchi space here. In
\cref{the-change-of-probability-with-a-mass-at-0} we will see how to
change the probability measure in order to lift Calls on Calls to Calls
on Calls with no mass in \(0\), and will provide the connection with the
generalized transformation of this section.

	To introduce the generalized transformation, firstly consider
\(\alpha\geq 0\) and the quantity \(\mathbb V_{k,\alpha}\) defined as
follows.

\begin{definition}

For any $c \in \mathbb C$ and $\alpha, k\geq0$ with $c(k)>0$, we define the transformation
$$\mathbb V_{k,\alpha} c(\cdot) := \frac{c(k + \alpha\cdot)}{c(k)}.$$

\end{definition}

	The transformation \(\mathbb T_k\) can be written as a function of
\(\mathbb V_{k,\alpha}\) in the sense that
\(\mathbb T_kc = \mathbb V_{k,c(k)}c\).

	It is easy to see that functions \(m\to\mathbb V_{k,\alpha} c(\cdot)\)
are convex and bounded by \(1\). Adding the requirement that
\(\alpha(k)\leq -\frac{c(k)}{c'(k)}\) also guarantees the lower bound
\((1-m)_+\), so that the transformations \(\mathbb V_{k,\alpha}\) can be
viewed as operating on the Tehranchi space.

\begin{lemma}[Properties of $\mathbb V_{k,\alpha}$]\label{lemmaV}

For any $c \in \mathbb C$, $\alpha, k\geq0$ with $c(k)>0$ and $\alpha\leq-\frac{c(k)}{c'(k)}$, it holds
\begin{enumerate}
\item $\mathbb V_{k,\alpha} c\in\mathbb C$;
\item $\frac{d}{dm} \mathbb V_{k,\alpha}c(0_+) = \frac{c'(k)}{c(k)}\alpha\geq 1$;
\item $\mathbb V_{k,\alpha}c(\infty) = \frac{c(\infty)}{c(k)}$;
\item $\mathbb V_{0,1} c=c$.
\end{enumerate}

\end{lemma}

\begin{proof}

The derivative and second derivative of $\mathbb V_{k,\alpha} c(m)$ with respect to $m$ are respectively $\frac{c'(k+\alpha m)}{c(k)}\alpha$ and $\frac{c'(k+\alpha m)}{c(k)}\alpha^2$. Then $\mathbb V_{k,\alpha} c(m)$ is convex in $m$. Since $c\in\mathbb C$, it is non-increasing and $\mathbb V_{k,\alpha} c(m)\leq 1$. The inequality $\mathbb V_{k,\alpha} c(m) \geq (1-m)_+$ amounts to $c(k+\alpha m)-c(k) \geq -c(k) m$ for $m<1$; by the mean value theorem the LHS writes $\alpha c'(k+\alpha u)m$ for some $u$ in $]0,m[$ where the derivative is negative. Since $c$ is convex, the latter quantity is larger than $\alpha c'(k)m$ and $\mathbb V_{k,\alpha} c(m) \geq (1-m)_+$ holds true as soon as $\alpha\leq -\frac{c(k)}{c'(k)}$.

The other points follow immediately.

\end{proof}

	As the transformation \(\mathbb T_k\), also its generalization
\(\mathbb V_{k,\alpha}\) satisfies a semigroup property and, as a
consequence, iterations of this transformation do not further enrich the
family \(\{\mathbb V_{k,\alpha}: k \geq 0, \alpha\geq0\}\):

	\begin{lemma}[Iterates of $\mathbb V_{k,\alpha}$]\label{lemmaIteratesV}

It holds $\mathbb V_{b,\beta} \mathbb V_{a,\alpha} c = \mathbb V_{a+\alpha b,\alpha\beta}c$.

Furthermore, if $\alpha\leq-\frac{c(a)}{c'(a)}$ and $\beta\leq-\frac{\mathbb V_{a,\alpha} c(b)}{\frac{d}{dm}\mathbb V_{a,\alpha} c(b)}$ then $\alpha\beta\leq-\frac{c(a+\alpha b)}{c'(a+\alpha b)}$.

\end{lemma}

\begin{proof}

The proof of the first statement is similar to the proof of \cref{lemmaTSemigroup}. 

Since $\frac{d}{dm}\mathbb V_{a,\alpha} c(b)=\frac{c'(a+\alpha b)}{c(a)}\alpha$, if $\beta\leq-\frac{\mathbb V_{a,\alpha} c(b)}{\frac{d}{dm}\mathbb V_{a,\alpha} c(b)}$ then $\beta\leq-\frac{c(a+\alpha b)}{\alpha c'(a+\alpha b)}$ and the second statement follows.

\end{proof}

	The second statement of \cref{lemmaIteratesV} implies that the family
\(\bigl\{\mathbb V_{k,\alpha}c: c\in\mathbb C, k \geq 0, 0\leq\alpha\leq-\frac{c(k)}{c'(k)}\bigr\}\)
(where \(c\in\mathbb C\) is also a parameter) is closed under
iterations.

	From \cref{lemmaV}, we see that the critical case
\(\alpha = -\frac{c(k)}{c'(k)}\) is of particular interest since it will
entail the property \(\frac{d}{dm} \mathbb V_{k,\alpha}c(0_+) = -1\).
This gives rise to a new transform on the Tehranchi space:

	\begin{definition}
For any $c \in \mathbb C$ and $k\geq0$ with $c(k)>0$ and $c'(k)\neq 0$ we define the transformation
$$\mathbb U_k c(\cdot) := \frac{c\bigl(k - \frac{c(k)}{c'(k)} \cdot\bigr)}{c(k)}.$$

\end{definition}

The transformation \(\mathbb U_k\) can be written as a function of
\(\mathbb V_{k,\alpha}\) in the sense that
\(\mathbb U_k c = \mathbb V_{k,-\frac{c(k)}{c'(k)}}c\), so properties of
the latter transformation (for fixed \(c\)) still hold for the former
one.

\begin{lemma}[Properties of $\mathbb U_k$]\label{lemmaU}

For any $c \in \mathbb C$ and $k\geq0$ with $c(k)>0$ and $c'(k)\neq 0$, it holds:
\begin{enumerate}
\item $\mathbb U_k c \in \mathbb C$;
\item $\frac{d}{dm} \mathbb U_kc(0_+) = -1$;
\item $\mathbb U_kc(\infty) = \frac{c(\infty)}{c(k)}$;
\item If $c'(0_+)=-1$, then $\mathbb U_0 c=c$.
\end{enumerate}

\end{lemma}

	\Cref{lemmaIteratesV} can be applied to \(\mathbb U_k\) but it does not
automatically guarantee that iterates of \(\mathbb U_k\) are still
functions in the family \(\{\mathbb U_k c: k \geq 0 \}\), even though
they certainly live in
\(\{\mathbb V_{k,\alpha}: k \geq 0, \alpha\geq0\}\). In the following
lemma we prove this point.

\begin{lemma}[Iterates of $\mathbb U_k$]

It holds $\mathbb U_b \mathbb U_a c= \mathbb U_{a-\frac{c(a)}{c'(a)}b}c$.

\end{lemma}

\begin{proof}

The proof can be shown directly as in \cref{lemmaTSemigroup}. Alternatively, applying \cref{lemmaIteratesV}, we have $\mathbb U_b \mathbb U_a c = \mathbb V_{b,\beta} \mathbb V_{a,\alpha} c$ where $\alpha = -\frac{c(a)}{c'(a)}$ and $\beta=-\frac{\mathbb V_{a,\alpha}c(b)}{\frac{d}{dm}\mathbb V_{a,\alpha}c(b)} = -\frac{c(a+\alpha b)}{\alpha c'(a+\alpha b)}$. Then $a+\alpha b = a-\frac{c(a)}{c'(a)}b$ and $\alpha\beta=-\frac{c(a+\alpha b)}{c'(a+\alpha b)}$, so $\mathbb V_{a+\alpha b,\alpha\beta}c = \mathbb U_{a-\frac{c(a)}{c'(a)}b}c$.

\end{proof}

	The transformation \(\mathbb U_k\) has a derivative in \(0_+\) equal to
\(-1\) and is then linked to probabilitiees with no mass in \(0\). This
will allow us to define new closed pricing formulas in
\cref{smile-symmetry-and-a-lift-of-the-relative-pricing-function}, that
we will call lifted Calls on Calls.

	\hypertarget{new-closed-formulas}{%
\section{New closed formulas}\label{new-closed-formulas}}

	In this section, we provide a quasi-closed formula for the pricing
function within the Carr-Pelts-Tehranchi family (see
\cite{carr2015duality,tehranchi2020black}), which generalizes the
Black-Scholes pricing function associated to the standard normal density
to any log-concave (and even, unimodal, as shown by Vladimir Lucic in
\cite{lucic2020lecture}) density function. This includes the
Black-Scholes case as a particular case. Those pricing functions are the
pricing functions of options on option, where the price of the latter
option is viewed as the underlier.

The reason to work with this family of pricing function is that a
variational formula for the option price, reminiscent of a dual
transform, is available, and it turns out that this variational formula
can be inverted to get an expression for the underlier value in terms of
the option price and the other parameters.

In the second section below, we derive new closed formulas from the
normalized price transformations - these formulas will yield (new)
homogeneous pricing functions when de-normalized.

	\hypertarget{the-carr-pelts-tehranchi-family}{%
\subsection{The Carr-Pelts-Tehranchi
family}\label{the-carr-pelts-tehranchi-family}}

	Knowing the expression of the underlier \(S\) as a function of the
option price \(X:=C(S,K)\) yields a closed formula for the option price
which is given by \(C(S,K+M)\), as a pricing function of \(X\) and
\(M\). In general, such an expression is unavailable, even if one can
resort to straightforward numerical procedures like a basic dichotomy to
compute it numerically, given the monotonicity of the map
\(S \to C(S,K)\).

It turns out that one can say more in the case of the
Carr-Pelts-Tehranchi family, due to the availability of a particular
variational formulation for the option price.

	We dub Carr--Pelts--Tehranchi (CPT) model the explicit arbitrage-free
parametrization for FX option prices introduced by Carr and Pelts in
\(2015\) at a conference in honor of Steven Shreve at Purdue university
(see \cite{carr2015duality}). The model has then been independently
rediscovered by Tehranchi in \cite{tehranchi2020black} while studying
advanced properties of the Black--Scholes formula.

In the CPT model, the family of Call prices is indexed by log-concave
densities \(f:\mathbb R\to[0,\infty[\) and increasing functions
\(y:[0,\infty[\to\mathbb R\) (which correspond to the total implied
volatility in the Black-Scholes framework). The Black-Scholes model is a
special case of CPT choosing \(f\) to be the standard normal probability
density function \(\phi\) and \(y(t) = v = \sigma\sqrt t\) with
reference to \cref{eqBS}.

Similarly to Black-Scholes, the CPT model has the nice feature that
option prices have a closed quasi-explicit formula. Indeed, the CPT Call
price is \begin{equation}\label{eqCPTC}
C^\text{CPT}(S,K;f,y(t)) := \int_{-\infty}^{\infty}\biggl(Sf(z+y(t))-Kf(z)\biggr)_+\,dz.
\end{equation} Tehranchi shows in section 3.2 of
\cite{tehranchi2020black} that if \(f\) is log-concave and \(y\) is
increasing, then prices in \cref{eqCPTC} represent a Call price surface
of the form \(E[(S_T-K)_+]\) for a certain non-negative supermartingale
\(S_t\) such that \(E[S_T]=S\). Equivalently, Call prices are
non-decreasing in \(t\), convex in \(K\) and equal to \((S-K)_+\) for
\(t=0\).

Remarkably, prices in \cref{eqCPTC} can actually be represented with a
formulation very close to the Black-Scholes one as
\[C^\text{CPT}(S,K;f,y(t)) = SF(d(K,y(t);f)+y(t)) - KF(d(K,y(t);f)\]
where \[d(K,y;f) := \sup\biggl\{z:\frac{f(z+y)}{f(z)}\geq K\biggr\}\]
and \(z\) lives in the support of \(f\), and \(F\) is the cumulative
density function associated with \(f\). In the Black-Scholes case, the
function \(d(K,y(t);f)\) can be obtained explicitly and is given by the
classical expression
\(d_2(S,K,v)=-\frac{\log\frac{K}{S}}{v} - \frac{v}{2}\).

Note that the CPT pricing functions are homogeneous ones.

Furthermore, Lucic has shown in \cite{lucic2020lecture} that under the
more general hypothesis that \(f\) is unimodal, i.e.~it has a single
peak (point of maximum), and \(y\) is increasing, prices in
\cref{eqCPTC} are still a Call price surface.

Since in the present article we are considering smiles of the Call
surface, i.e.~for fixed time-to-maturity, we will drop the dependence to
\(t\) of \(y\).

	One of the important properties of the CPT family is the availability of
a variational formula for the option price (Theorem 4.1.2 of
\cite{tehranchi2020black}):
\[C^\text{CPT}(S,K;f,y(t)) = \sup_{p \in ]0,1[}  S F(F^{-1}(p)+y) - p K.\]

	This formula is the key of the following result.

\begin{lemma}[Inversion of the CPT formula]\label{lemmaPsi}

Let $f$ a unimodal probability density function and $F$ its cumulative density function. For $K,y\geq 0$ let $X=C^\text{CPT}(S,K;f,y)$, then it holds
$$S = K \psi\biggl(\frac{X}K,y;F\biggr)$$
where
$$\psi(a,y;F):= \inf_{r \in \mathbb R} \frac{a+F(r-y)}{F(r)}.$$

\end{lemma}

	\begin{proof}
It holds
$$X = \sup_{p \in ]0,1[}  S F(F^{-1}(p)+y) - p K$$
so that for every $p$, $X \geq  S F(F^{-1}(p)+y) - p K$ or yet $S \leq \frac{X+pK}{F(F^{-1}(p)+y)}$, and also
$$S = \inf_{p \in ]0,1[}  \frac{X+pK}{F(F^{-1}(p)+y)}.$$
Set $r:=F^{-1}(p)+y$, then $p=F(r-y)$ and, given that the range of $r$ when $p$ runs into $]0,1[$ is $\mathbb R$ irrespective of $y$, the conclusion follows.

\end{proof}

	We can easily apply the result from this lemma to the relation
\(\hat C_K(C(S,K),M)=C(S,K+M)\) and obtain the following.

\begin{proposition}[Quasi-closed formula for the Call on Call pricing function in the CPT family]\label{propQuasiClosedCPT}

Let $f$ a unimodal probability density function. For $K,M,y\geq 0$ it holds
\begin{align*}
\hat C_K(X,M) &= \int_{-\infty}^{\infty}\biggl(K\psi\biggl(\frac{X}K,y;F\biggr)f(z+y)-(K+M)f(z)\biggr)_+\,dz\\
&= K\psi\biggl(\frac{X}K,y;F\biggr)F(d(K+M,y;f)+y) - (K+M)F(d(K+M,y;f)).
\end{align*}

In particular, in the Black-Scholes case with $v=\sigma\sqrt T$
\begin{align*}
\hat C_K(X,M) =&\ K\psi\biggl(\frac{X}K,v;\Phi\biggr)\Phi\biggl(d_1\Bigl(K\psi\Bigl(\frac{X}{K},v;\Phi\Bigr),K+M,v\Bigr)\Bigr) +\\
&- (K+M)\Phi\biggl(d_2\Bigl(K\psi\Bigl(\frac{X}{K},v;\Phi\Bigr),K+M,v\Bigr)\biggr)
\end{align*}
where $d_{1,2}(a,b,v)=-\frac{\log\frac{b}{a}}{v} \pm \frac{v}{2}$.

\end{proposition}

	Observe that in the Black-Scholes case of the above proposition,
choosing \(M=0\), we find the expression \begin{align*}
\hat C_K(X,0) = K\psi\biggl(\frac{X}K,v;\Phi\biggr)\Phi\biggl(d_1\Bigl(K\psi\Bigl(\frac{X}{K},v;\Phi\Bigr),K,v\Bigr)\Bigr) - K\Phi\biggl(d_2\Bigl(K\psi\Bigl(\frac{X}{K},v;\Phi\Bigr),K,v\Bigr)\biggr)
\end{align*} which is the classic Black-Scholes formula for the Call
\(C(\tilde S,K)\) where \(\tilde S=K\psi\bigl(\frac{X}K,v;\Phi\bigr)\).
This was indeed expected since Call prices with null strike coincide
with the value of their underlier. Furthermore, by the definition of
\(\psi\) in \cref{lemmaPsi}, the underlier of the Call option \(X\) with
strike \(K\) is \(S=K\psi\bigl(\frac{X}K,v;\Phi\bigr)\), so that
\(\tilde S=S\) and the above expression is the Call price of an option
with strike \(K\) and underlier \(S\), i.e.~it coincides with \(X\):
\[X = K\psi\biggl(\frac{X}K,v;\Phi\biggr)\Phi\biggl(d_1\Bigl(K\psi\Bigl(\frac{X}{K},v;\Phi\Bigr),K,v\Bigr)\Bigr) - K\Phi\biggl(d_2\Bigl(K\psi\Bigl(\frac{X}{K},v;\Phi\Bigr),K,v\Bigr)\biggr).\]

	\hypertarget{numerical-computation-of-psi}{%
\subsubsection{\texorpdfstring{Numerical computation of
\(\psi\)}{Numerical computation of \textbackslash psi}}\label{numerical-computation-of-psi}}

	The function \(\psi\) in \cref{lemmaPsi} is still not explicit. However,
one can study more precisely \(\psi\) in the case of a strictly
log-concave \(f\), which covers the Black-Scholes case in particular.

Indeed, let us look at the function \(\psi(a, y)\). We consider \(a>0\)
and call \[\gamma(r;a,y):=\frac{a+F(r-y)}{F(r)}\] the argument of the
infimum. Then \(\gamma(\infty;a,y)=1+a\) and \(\gamma\) is
non-increasing iff
\(\gamma'(r;a,y)=\frac{F(r)f(r-y)-f(r)(a+F(r-y))}{F(r)^2}\leq0\), or iff
\(a\geq\delta(r;y)\) where
\[\delta(r;y):=\frac{f(r-y)}{f(r)}F(r)-F(r-y).\]

We have the following:

	\begin{lemma}\label{lemmaMinPsi}

Let $f$ a strictly log-concave probability density function and $F$ its cumulative density function, and let $a,y>0$. Then $\delta(r;y)$ is strictly increasing in $r$ and:

\begin{itemize}
\item If $\delta(\infty;y)\leq a$ then $\psi(a,y)=1+a$;
\item If $\delta(\infty;y)>a$ then $\psi(a,y) = \frac{a+F(r^*-y)}{F(r^*)} = \frac{f(r^*-y)}{f(r^*)}$ where $r^*:=\delta^{-1}(a;y)$. 
\item In the Black-Scholes case $\delta(\infty;v)=\infty$ for every $v$.
\end{itemize}

\end{lemma}

	\begin{proof}

Observe that if $\gamma$ has a finite limit at $-\infty$, then $\delta(r;y)=F(r)\bigl(\frac{f(r-y)}{f(r)}-\frac{F(r-y)}{F(r)}\bigr)$ has a limit equal to $0$ due to l'H�pital's rule. Also, if $\gamma$ explodes at $-\infty$, then its derivative must be non-positive at $-\infty$, i.e. $a$ is always larger or equal than $\delta(-\infty;y)$, which does not depend on $a$. As a consequence, $\delta(-\infty;y)=0$ in any case.

If $\delta(r;y)$ is increasing, two scenarios are possible:
\begin{itemize}
    \item $\delta(\infty;y)>a$, then $a$ cannot be always larger than $\delta(r,y)$ and the function $\gamma$ has at least one point of minimum. Also, since $\delta$ is monotonous in $r$, there is a unique $r^*$ such that $a=\delta(r^*;y)$ and this point is also the point of minimum of $\gamma$, i.e. $\psi(a,y) = \frac{a+F(r^*-y)}{F(r^*)} = \frac{f(r^*-y)}{f(r^*)}$;
    \item $\delta(\infty;y)\leq a$ (in particular $\delta$ is not surjective), then $\gamma$ is decreasing and therefore $\psi(a,y)=1+a$.
\end{itemize}

Now since $f$ is strictly log-concave, then the function $\frac{f'}{f}$ is decreasing. The derivative of $\delta$ with respect to $r$ is $\frac{F(r)}{f(r)^2}(f(r)f'(r-y)-f(r-y)f'(r))$ and this is positive iff $\frac{f'(r-y)}{f(r-y)}>\frac{f'(r)}{f(r)}$, which holds true.

Then, if $\delta(\infty;y)>a$ there exists a unique $r^*:=\delta^{-1}(a;y)$ and $\psi(a,y)= \frac{a+F(r^*-y)}{F(r^*)}$. Otherwise, if $\delta(\infty;y)\leq a$, it holds $\psi(a,y)=1+a$.

In the Black-Scholes case, $f$ is strictly log-concave and $\frac{f(r-v)}{f(r)} = \exp\bigl(-\frac{(r-v)^2-r^2}{2}\bigr) = \exp\bigl(\frac{v(2r-v)}{2}\bigr)$ which explodes for $r$ going to $\infty$. Then $\delta(\infty;v)=\infty>a$.

\end{proof}

	We can actually give more details on the bounds of the point of minimum
\(r^*\) of the function \(\gamma\), when it exists (i.e.~when
\(\delta(\infty;y)>a\)). In the following lemma we find a lower bound
\(r^l\) and an upper bound \(r^u\) for \(r^*\) under the hypothesis of
surjectivity of the function \(\frac{f'}{f}\), and in the successive
corollary we study the Black-Scholes case, where the bounds are actually
explicit. In this way, the point \(r^*\) can be numerically computed in
a very quick way inverting the function \(\delta\) in the interval
\([r^l,r^u]\).

	\begin{lemma}\label{lemmaBoundr}

Let $f$ a strictly log-concave probability density function and $F$ its cumulative density function, and let $a,y>0$. Then $\frac{f'(r)}{f(r)}$ is decreasing and $f$ is unimodal. Let $s$ the unique point of maximum of $f$.

In the case $\delta(\infty;y)>a$, if $\frac{f'(r)}{f(r)}$ is surjective, it holds

\begin{itemize}
\item If $\delta(s;y)\geq a$ then $\tilde r<r^*\leq s$, where $\tilde r$ is the unique $r\leq s$ solving $a=F(r)-F(r-y)$.
\item If $\delta(s;y)<a$ and $\delta(s+y;y)\geq a$ then $s<r^*\leq s+y$.
\item If $\delta(s+y;y)<a$ then $s+y<r^*<\hat r$, where $\hat r$ is the unique $r$ solving $\frac{f(r-y)}{f(r)}=\frac{a}{F(s)}+1$.
\end{itemize}

\end{lemma}

	\begin{proof}

Firstly observe that since $f$ is strictly log-concave, then the function $r\to\frac{f'(r)}{f(r)}$ is decreasing while the function $r\to\frac{f(r-y)}{f(r)}$ is increasing, given that its derivative is $\frac{f'(r-y)f(r)-f(r-y)f'(r)}{f(r)^2}$. Secondly, from the proof of Theorem 4.1.6. of \cite{tehranchi2020black}, it holds
$$\frac{f'(r)}{f(r)}\leq\frac{1}{y}\log\frac{f(r)}{f(r-y)}\leq\frac{f'(r-y)}{f(r-y)}.$$
Then if $\frac{f'(r)}{f(r)}$ goes to $-\infty$ at $\infty$, the function $\frac{f(r-y)}{f(r)}$ explodes at $\infty$, while if $\frac{f'(r)}{f(r)}$ goes to $\infty$ at $-\infty$, then $\frac{f(r-y)}{f(r)}$ goes to $0$ at $-\infty$.

In the case $\delta(\infty;y)>a$, from \cref{lemmaMinPsi} there exists a unique $r^*$ such that $a=\delta(r^*;y)$. Since $F(r)>F(r-y)$, it holds
$$a=\delta(r^*;y)>\biggl(\frac{f(r^*-y)}{f(r^*)}-1\biggr)F(r^*-y).$$
If $\delta(s;y)\geq a$ then $r^*\leq s$. Otherwise $r^*>s$. If $\delta(s+y;y)\geq a$ then $r^*\leq s+y$. Otherwise $r^*>s+y$, so $F(r^*-y)>F(s)$, $\frac{f(r^*-y)}{f(r^*)}>1$ and $\frac{a}{F(s)}+1>\frac{f(r^*-y)}{f(r^*)}$. Since the function $r\to\frac{f(r-y)}{f(r)}$ is increasing and explodes at $\infty$ under the hypothesis that $\frac{f'(r)}{f(r)}$ goes to $-\infty$, then there exists a unique $\hat r$ such that $\frac{f(\hat r-y)}{f(\hat r)}=\frac{a}{F(0)}+1$. Furthermore, $r^*<\hat r$.

In the previous step we have already found a lower bound (either $s$ or $s+y$) in the case $\delta(s;y)<a$. If $\delta(s;y)\geq a$, then $r^*\leq s$ and $\frac{f(r^*-y)}{f(r^*)}<1$, so
$$a=\delta(r^*;y)<F(r^*)-F(r^*-y).$$
The function $r\to F(r)-F(r-y)$ has derivative $f(r)-f(r-y)$, which is positive for $r\leq s$. Then there is a unique $\tilde r\leq s$ solving $a=F(\tilde r)-F(\tilde r-y)$ and $\tilde r<r^*$.

\end{proof}

	\Cref{lemmaBoundr} can be further exploit in the Black-Scholes case. It
turns out that the bounds for \(r^*\) are explicit and do not need any
inversion algorithm.

	\begin{corollary}[Explicit bounds for $r^*$ in the Black-Scholes case]\label{corolBoundrBS}

In the Black-Scholes case:

\begin{itemize}
\item If $a\leq\sqrt{\frac{\pi}{2}}\phi(v) + \Phi(v)-1$ then $-\sqrt{-2\log\bigl(\frac{a}{v}\sqrt{2\pi}\bigr)}<r^*\leq 0$.
\item If $\sqrt{\frac{\pi}{2}}\phi(v) + \Phi(v)-1<a\leq \frac{1}{\sqrt{2\pi}}\frac{\Phi(v)}{\phi(v)}-\frac{1}{2}$ then $0<r^*\leq v$.
\item If $a>\frac{1}{\sqrt{2\pi}}\frac{\Phi(v)}{\phi(v)}-\frac{1}{2}$ then $v<r^*<\frac{1}{2}\bigl(v+\frac{2}{v}\log(2a+1)\bigr)$.
\end{itemize}

\end{corollary}

	\begin{proof}

In the Black-Scholes case, $f=\phi$, $s=0$ and $\phi'(r)=-r\phi(r)$, so that $\frac{\phi'(r)}{\phi(r)}=-r$ which is surjective. In \cref{lemmaMinPsi} we showed that $\delta(\infty;v)=\infty>a$, so three scenarios are possible applying \cref{lemmaBoundr}.

In the first scenario, the condition $\delta(s;v)\geq a$ reads $\frac{\phi(-v)}{\phi(0)}\Phi(0)-\Phi(-v)\geq a$. Since $\Phi(-v)=1-\Phi(v)$, $\phi(-v)=\phi(v)$, $\Phi(0)=\frac{1}{2}$ and $\phi(0)=\frac{1}{\sqrt{2\pi}}$, the condition is equivalent to $a\leq\sqrt{\frac{\pi}{2}}\phi(v) + \Phi(v)-1$. In this case the point $\tilde r\leq0$ is the only $r$ satisfying $a=\int_{r-v}^{r} \phi(z)\,dz$ and it holds $r^*>\tilde r$. Since $\tilde r\leq 0$, it holds $\int_{r-v}^{r} \phi(z)\,dz<\int_{r-v}^{r} \phi(\tilde r)\,dz = \phi(\tilde r)v$, so that $\frac{a}{v}\sqrt{2\pi}<\exp\bigl(-\frac{\tilde r^2}{2}\bigr)$. As a consequence, the LHS is smaller than $1$ and we can solve $\tilde r^2<-2\log\bigl(\frac{a}{v}\sqrt{2\pi}\bigr)$, which implies in particular $\tilde r>-\sqrt{-2\log\bigl(\frac{a}{v}\sqrt{2\pi}\bigr)}$.

In the second scenario, the condition $\delta(s+v;v)\geq a$ is $\frac{\phi(0)}{\phi(v)}\Phi(v) - \Phi(0)\geq a$, or equivalently $a\leq \frac{1}{\sqrt{2\pi}}\frac{\Phi(v)}{\phi(v)}-\frac{1}{2}$.

In the last scenario, the point $\hat r$ solves $\frac{\phi(\hat r-v)}{\phi(\hat r)}=2a+1$. The LHS is $\exp\bigl(\frac{v(2r-v)}{2}\bigr)$ and solving we find $\hat r = \frac{1}{2}\bigl(v+\frac{2}{v}\log(2a+1)\bigr)$.

\end{proof}

	Thanks to \cref{lemmaBoundr} and \cref{corolBoundrBS} we have found
specific bounds for \(r^*\). Then extremely fast numerical
implementations based on standard methods such as the \verb+brentq+
function of the \verb+scipy.optimize+ library in Python can be obtained
using these bounds.

	\hypertarget{formulas-from-normalized-transformations}{%
\subsection{Formulas from normalized
transformations}\label{formulas-from-normalized-transformations}}

	The transformation \(\mathbb V_{k,\alpha}\) of
\cref{a-slight-generalization}, and so \(\mathbb T_k\) and
\(\mathbb U_k\), allow to generate new pricing formulas using the
following recipe: start from a Call pricing function with closed
formula, normalize it, apply the transformation and de-normalize to get
another closed formula. This allows to extend any closed formula to a
\(2\)-parameter family of closed formulas.

In other words, we look at the pricing formula in the new world, but
consider eventually applying it to the usual world: we take a financial
engineer point of view here, where any pricing function depending on
some parameters is a useful candidate to calibrate the current state of
the market (in the usual world).

So we go through the following pipeline of transformations:

\begin{enumerate}
\def\labelenumi{\arabic{enumi}.}
\tightlist
\item
  Start from any arbitrage-free Call pricing function \(K \to C(S,K)\);
\item
  Normalize it and get a function \(c \in \mathbb C\) defined as
  \(c(k) := \frac{C(S,kS)}{S}\);
\item
  Apply, for any \(k,\alpha\geq0\) such that \(c(k)>0\) and
  \(\alpha\leq-\frac{c(k)}{c'(k)}\), the transformation
  \(\mathbb V_{k,\alpha}\) to \(c\);
\item
  Get a new arbitrage-free Call pricing function given by the formula
  \(M \to X \mathbb V_{k,\alpha} c\bigl(\frac{M}{X}\bigr)\), where
  \(X>0\) represents the value of the new underlier.
\end{enumerate}

The above steps can be applied in particular for the two special choices
of \(\alpha\) defining \(\mathbb T_k\) and \(\mathbb U_k\):
\(\alpha=c(k)\) and \(\alpha=-\frac{c(k)}{c'(k)}\).

Observe that if we choose \(k=\frac{K}{S}\) and \(X=C(S,K)\) the new
Call pricing function obtained using the transformation
\(\mathbb V_{k,c(k)}=\mathbb T_k\) coincides with
\(M\to\hat C_K(C(S,K),M)=C(S,K+M)\).

	Let us compute the new closed formulas we obtain implementing the above
pipeline for the known families of Black-Scholes, SVI (composed with the
Black-Scholes pricing function), and CPT, that all provide closed-form
formulas.

	\hypertarget{a-2-parameter-enrichment-of-the-black-scholes-formula}{%
\subsubsection{A 2-parameter enrichment of the Black-Scholes
formula}\label{a-2-parameter-enrichment-of-the-black-scholes-formula}}

	Let \(\text{BS}(S,K)\) the Black-Scholes function defined in
\cref{eqBS}. Then
\(\text{BS}(S,K) = S \text{bs}\bigl(\frac{K}{S}\bigr)\) where the
normalized Black-Scholes pricing function \(\text{bs}\) belongs to the
Tehranchi space. We can therefore consider the two families of functions
indexed by \(k\) \begin{align*}
\text{bs}^{\mathbb T}_k\Bigl(\frac{M}{X}\Bigr) &:= \frac{\text{bs}\bigl(k+\text{bs}(k)\frac{M}{X}\bigr)}{\text{bs}(k)}\\
\text{bs}^{\mathbb U}_k\Bigl(\frac{M}{X}\Bigr) &:= \frac{\text{bs}\bigl(k-\frac{\text{bs}(k)}{\text{bs}'(k)}\frac{M}{X}\bigr)}{\text{bs}(k)}
\end{align*} leading to the \emph{enriched Black-Scholes models}
\begin{align*}
\text{BS}^{\mathbb T}_k(X,M) &:= X \text{bs}^{\mathbb T}_k\Bigl(\frac{M}{X}\Bigr) = X\frac{\text{bs}\bigl(k+\text{bs}(k)\frac{M}{X}\bigr)}{\text{bs}(k)}\\
\text{BS}^{\mathbb U}_k(X,M) &:= X \text{bs}^{\mathbb U}_k\Bigl(\frac{M}{X}\Bigr) = X\frac{\text{bs}\bigl(k-\frac{\text{bs}(k)}{\text{bs}'(k)}\frac{M}{X}\bigr)}{\text{bs}(k)}.\\
\end{align*}

	In \cref{figureBSSigma} we plot the \(\text{BS}^{\mathbb T}_k(X,M)\) and
\(\text{BS}^{\mathbb U}_k(X,M)\) prices (for fixed maturity) for
different values of \(k\) and a fixed value of \(X=1\). The original
implied total volatility is fixed at \(0.2\).

\begin{figure}
	\centering
	\includegraphics[width=.8\linewidth]{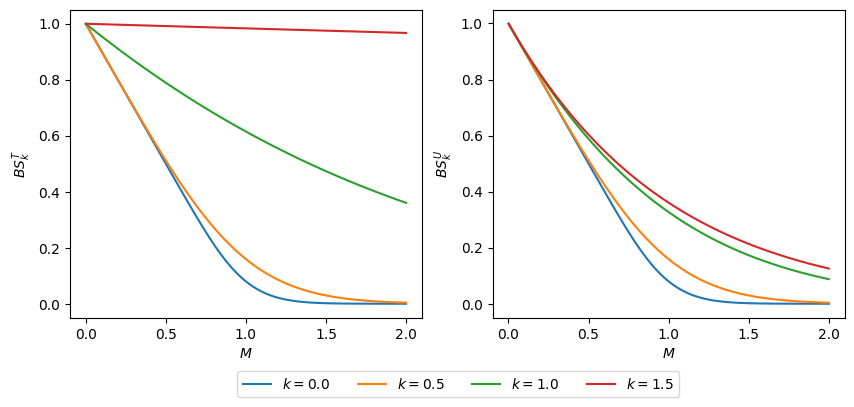}
	\caption{Prices $\text{BS}^{\mathbb T}_k(X,M)$ and $\text{BS}^{\mathbb U}_k(X,M)$ obtained from Black-Scholes formula applying transformations $\mathbb T_k$ and $\mathbb U_k$ respectively. The original Black-Scholes implied total volatility is set at $0.2$.}
	\label{figureBSSigma}
\end{figure}
    
	In order to define the \(2\)-parameter enrichment of the Black-Scholes
formula \(\text{BS}^{\mathbb V}_{k,z}\) as suggested in the introduction
of this section, let us make explicit the second parameter of the
function \(\text{bs}\), the total implied volatility
\(v=\sigma\sqrt T\). As seen in \cref{a-slight-generalization}, the
families \(\text{BS}^{\mathbb T}_k\) and \(\text{BS}^{\mathbb U}_k\) are
a particular choice in the more generic set of families
\(X\frac{\text{bs}\bigl(k+\alpha\frac{M}{X},v\bigr)}{\text{bs}(k,v)}\).
In particular, \(\text{BS}^{\mathbb T}_k\) corresponds to the choice
\(\alpha=\text{bs}(k)\) and \(\text{BS}^{\mathbb U}_k\) to the choice
\(\alpha=-\frac{\text{bs}(k)}{\text{bs}'(k)}\).

More generally, we can represent \(\alpha\) as the value of a normalized
Black-Scholes pricing function at \(k\): \(\alpha=\text{bs}(k,z)\) where
the variability is given by the choice of the implied total volatility
\(z\). We have then identified a \emph{2-parameters enrichment} of the
Black-Scholes pricing function that satisfies properties of
\cref{propEasyCallsOnCalls}:

\begin{proposition}[$2$-parameter enrichment of the Black-Scholes formula from the normalized pricing function]

For any $k\geq0$ and $z>0$ we define the 2-parameter enrichment of the Black-Scholes pricing function with $v=\sigma_1\sqrt{T}$ and $z=\sigma_2\sqrt{T}$
$$\text{BS}^{\mathbb V}_{k,z}(X,M,v) := X \frac{\text{bs}\bigl(k+ \text{bs}(k, z)\frac{M}{X}, v\bigr)}{\text{bs}(k, v)}.$$

Then $\text{BS}^{\mathbb V}_{k,z}(X,M,v)$ is a non-increasing and convex function of $M$ with $\text{BS}^{\mathbb V}_{k,z}(X,M,v) \leq X$. If $z\leq v$, then $(X-M)_+\leq \text{BS}^{\mathbb V}_{k,z}(X,M,v)$.

In the case $z=v$, the function $k\to\text{BS}^{\mathbb V}_{k,v}(X,M,v)$ is non-decreasing.

In particular $\text{BS}^{\mathbb V}_{0,z}(X,M,v) = X\text{bs}\bigl(\frac{M}{X}, v\bigr) = \text{BS}(X,M,v)$.

\end{proposition}

	\begin{proof}

The monotonicity and convexity with respect to $M$ are easy to be proved. Since $k+ \text{bs}(k, z)\frac{M}{X} \geq k$ and the function $\text{bs}(k,v)$ is decreasing in $k$, we obtain the first inequality. If $z\leq v$, then $\text{bs}(k,z)\leq\text{bs}(k,v)$ and $\text{BS}^{\mathbb V}_{k,z}(X,M,v)\geq X \frac{\text{bs}\bigl(k+ \text{bs}(k, v)\frac{M}{X}, v\bigr)}{\text{bs}(k,v)}$ which is larger than $(X-M)_+$ from \cref{lemmaT}.

Finally if $z=v$ we apply \cref{propTincreasing}.

\end{proof}

	\hypertarget{the-enriched-svi-models}{%
\subsubsection{The enriched SVI models}\label{the-enriched-svi-models}}

	The Stochastic Volatility Inspired model (SVI) is a model for the
implied total variance with formulation
\[\text{SVI}(h) = a+b(\rho(h-m)+\sqrt{(h-m)^2+\sigma^2})\] where
\(h=\log\frac{K}{S}\) is the log-forward moneyness. It was firstly
proposed by Jim Gatheral in 2004 at the Global Derivatives and Risk
Management Madrid conference \cite{gatheral2004parsimonious}.

In this model, Call prices at moneyness \(k\) are Black-Scholes prices
with implied total variance \(\sqrt{\text{SVI}(\log k)}\). We denote
these prices with
\[\text{BS}^\text{SVI}(S,K) := \text{BS}\biggl(S,K,\sqrt{\text{SVI}\biggl(\log \frac{K}{S}\biggr)}\biggr).\]
Applying \(\mathbb T_k\) and \(\mathbb U_k\), we obtain the
\emph{enriched SVI models} \begin{align*}
\text{BS}^{\text{SVI},\mathbb T}_k(X,M) &:= X\frac{\text{bs}^\text{SVI}\bigl(k+\text{bs}^\text{SVI}(k)\frac{M}{X}\bigr)}{\text{bs}^\text{SVI}(k)}\\
\text{BS}^{\text{SVI},\mathbb U}_k(X,M) &:= X\frac{\text{bs}^\text{SVI}\bigl(k-\frac{\text{bs}^\text{SVI}(k)}{(\text{bs}^\text{SVI})'(k)}\frac{M}{X}\bigr)}{\text{bs}^\text{SVI}(k)}.\\
\end{align*}

We plot prices \(\text{BS}^{\text{SVI},\mathbb T}_k(X,M)\) and
\(\text{BS}^{\text{SVI},\mathbb U}_k(X,M)\) of the enriched SVI models
in \cref{figureSVIhat}, where we choose as initial arbitrage-free
parameters for the SVI smile \(a = 0.01\), \(b = 0.1\), \(\rho = -0.6\),
\(m = -0.05\), and \(\sigma = 0.1\) (see Table 2 of
\cite{martini2022no}). We take a fixed value \(X=1\).

\begin{figure}
	\centering
	\includegraphics[width=.8\linewidth]{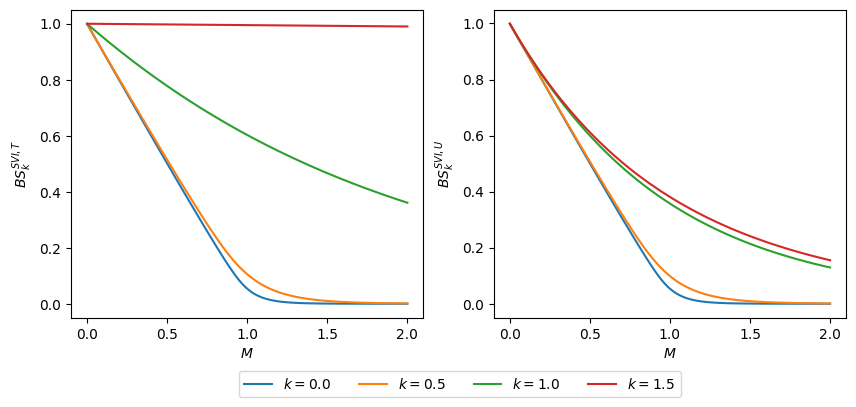}
	\caption{Prices $\text{BS}^{\text{SVI},\mathbb T}_k(X,M)$ and $\text{BS}^{\text{SVI},\mathbb U}_k(X,M)$ obtained from SVI prices applying transformations $\mathbb T_k$ and $\mathbb U_k$ respectively. The original arbitrage-free SVI parameters are $a = 0.01$, $b = 0.1$, $rho = -0.6$, $m = -0.05$, and $sigma = 0.1$.}
	\label{figureSVIhat}
\end{figure}
    
	Again, we can enrich also the SVI model to a \(2\)-parameter family
indexed by \(k\) and \(\alpha\):
\[\text{BS}^{\text{SVI},\mathbb V}_{k,\alpha}(X,M) := X\frac{\text{bs}^\text{SVI}\bigl(k+\alpha\frac{M}{X}\bigr)}{\text{bs}^\text{SVI}(k)}.\]

	\hypertarget{the-enriched-cpt-models}{%
\subsubsection{The enriched CPT models}\label{the-enriched-cpt-models}}

	The CPT prices are defined in \cref{eqCPTC} and have corresponding
normalized prices
\[c^\text{CPT}(k;f,y):=\frac{C^\text{CPT}(S,Sk;f,y)}{S}=\int_{-\infty}^{\infty}\bigl(f(z+y)-kf(z)\bigr)_+\,dz.\]

Even the CPT model can be extended via \(\mathbb T_k\) and
\(\mathbb U_k\) to get the \emph{enriched CPT models} \begin{align*}
C^{\text{CPT},\mathbb T}_k(X,M) &:= X\frac{c^\text{CPT}\bigl(k+c^\text{CPT}(k;f,y)\frac{M}{X};f,y\bigr)}{c^\text{CPT}(k;f,y)}\\
C^{\text{CPT},\mathbb U}_k(X,M) &:= X\frac{c^\text{CPT}\bigl(k-\frac{c^\text{CPT}(k;f,y)}{(c^\text{CPT})'(k;f,y)}\frac{M}{X};f,y\bigr)}{c^\text{CPT}(k;f,y)}\\
\end{align*} and via \(\mathbb V_{k,\alpha}\) to get the \(2\)-parameter
extension of the CPT model:
\[C^{\text{CPT},\mathbb V}_{k,\alpha}(X,M) := X\frac{c^\text{CPT}\bigl(k+\alpha\frac{M}{X};f,y\bigr)}{c^\text{CPT}(k;f,y)}.\]

	\hypertarget{smile-symmetry-and-a-lift-of-the-relative-pricing-function}{%
\section{Smile symmetry and a lift of the relative pricing
function}\label{smile-symmetry-and-a-lift-of-the-relative-pricing-function}}

	As already pointed out in \cref{remarkDeriv1}, Calls on Calls are
contracts written on underliers with a positive mass in \(0\), i.e.~that
can become null with positive probability. This implies some unusual
features such as a derivative of the pricing function with respect to
the strike which is larger than \(-1\) at the origin. However, it is
possible to change the probability measure in order to obtain new
contracts that do not present this feature anymore. Moreover, here is a
tight relationship with the symmetry operation applied to the smile, as
is well-known and detailed in section 2.2 of \cite{martini2021explicit}.
An analogous transformation is performed in section 2.2 of
\cite{tehranchi2020black} with the involution of Call prices.

	\hypertarget{the-change-of-probability-with-a-mass-at-0}{%
\subsection{The change of probability with a mass at
0}\label{the-change-of-probability-with-a-mass-at-0}}

	We firstly start with the definition of the probability \(P^*\)
associated to a non-negative non-constantly zero random variable
\(X_T\). Note that we re-introduce the maturity \(T\) here, in order to
convey some financial context, on one hand, and also to distinguish
those random variables from constant quantities known at the current
time (assumed to be time \(0\)).

\begin{definition}\label{defPstar}

Let $X_T$ a non-negative non-constantly zero random variable on the probability space $(P,\Omega)$, with finite expectation $E[X_T]$. We define
$$P^*(A):=\frac{E\bigl[\mathbbm{1}_AX_T\mathbbm{1}_{\{X_T>0\}}\bigr]}{E[X_T]}$$
for every $A\in\Omega$. We also denote $X:=E[X_T]$ and $P_0 := P(X_T=0)$.

\end{definition}

	An immediate consequence of the above definition is that \(P^*\) is
actually a probability measure on a subset of the original \(\Omega\).
The proof of the following lemma is elementary and so omitted.

\begin{lemma}\label{lemmaXstar}

$P^*$ is a probability measure on $\Omega^*=\{X_T>0\}$. Any random variable $Z$ on $(P,\Omega)$ can be restricted to a random variable on $(P^*,\Omega^*)$ (that we still denote with $Z$). Then $E^*[Z]=\frac{1}{X}E[ZX_T\mathbbm 1_{\{X_T>0\}}]$. 

Let $X^*_T:=\frac{1}{X_T}$. Then
$$E^*[X^*_T]=\frac{1-P_0}{X}.$$

\end{lemma}

	This lemma suggests to consider contracts under the probability \(P^*\)
on the underlier \(X_T^*=\frac{1}{X_T}\). Indeed, it holds
\begin{align*}
E^*\biggl[\biggl(\frac{1}{X_T}-\frac{1}K\biggr)_+\biggr] &= \frac{E\bigl[(K-X_T)_+\mathbbm{1}_{\{X_T>0\}}\bigr]}{XK}\\
E^*\Bigl[\Bigl(\frac{1}K-\frac{1}{X_T}\Bigr)_+\Bigr] &= \frac{E\bigl[(X_T-K)_+\bigr]}{XK}
\end{align*} which suggests that a Call price under \(P\) is also a Put
price on the underlier \(X_T^*\) under the new probability \(P^*\).
Furthermore
\(E[(K-X_T)_+] = E\bigl[(K-X_T)_+\mathbbm{1}_{\{X_T>0\}} + K\mathbbm{1}_{\{X_T=0\}}\bigr] = E\bigl[(K-X_T)_+\mathbbm{1}_{\{X_T>0\}}] + KP_0\)
so that the Put-Call-Parity
\[\frac{1-P_0}{X}-\frac{1}{K} = E^*\biggl[\biggl(\frac{1}{X_T}-\frac{1}{K}\biggr)_+\biggr]-E^*\biggl[\biggl(\frac{1}{K}-\frac{1}{X_T}\biggr)_+\biggr]\]
allows us to express the price of Calls on \(\frac{1}{X_T}\) with strike
\(\frac{1}K\) under \(P^*\) with respect to the original Call prices as
\begin{align*} 
E^*\biggl[\biggl(\frac{1}X_T-\frac{1}K\biggr)_+\biggr] &=  E^*\biggl[\biggl(\frac{1}K-\frac{1}{X_T}\biggr)_+\biggr]+\frac{1-P_0}{X}- \frac{1}K\\
&= \frac{E\bigl[(X_T-K)_+\bigr]}{XK}+\frac{1-P_0}{X}- \frac{1}K.\\
\end{align*}

	We can re-apply the same procedure to the underlier \(X_T^*\) under
\(P^*\), defining a probability measure \(P^{**}\) and an underlier
\(X_T^{**}\).

\begin{definition}

Let $X_T$ a non-negative non-constantly zero random variable on the probability space $(P,\Omega)$, with finite expectation $E[X_T]$. We define
$$P^{**}(A):=\frac{E^*\bigl[\mathbbm 1_AX_T^*\bigr]}{E^*[X_T^*]}$$
for every $A\in\Omega^*$.

\end{definition}

The random variable \(\frac{1}{X_T^*}\) is well defined and does not
vanish under \(P^*\) since it is defined on the space \(\Omega^*\), so
that the corresponding of \cref{lemmaXstar} for the function \(P^{**}\)
becomes the following.

\begin{lemma}\label{lemmaXstarstar}

$P^{**}$ is a probability measure on $\Omega^*=\{X_T>0\}$. Any random variable $Z$ on $(P,\Omega)$ can be restricted to a random variable on $(P^{**},\Omega^*)$ (that we still denote with $Z$). Then $E^{**}[Z]=\frac{X}{1-P_0}E^*[Z]=\frac{1}{1-P_0}E[Z\mathbbm 1_{\{X_T>0\}}]$. 

Let $X_T^{**}:=\frac{1}{X_T^*}$. Then
$$E^{**}[X_T^{**}]=\frac{X}{1-P_0}.$$

\end{lemma}

	We can now consider Call prices written on \(X_T^{**}\) using again the
Put-Call-Parity: \begin{align*} 
E^{**}\biggl[\biggl(\frac{1}{X_T^*}-\frac{1}K\biggr)_+\biggr] &=  E^{**}\biggl[\biggl(\frac{1}K-\frac{1}{X_T^*}\biggr)_+\biggr]+\frac{X}{1-P_0}- \frac{1}K\\
&= \frac{1}{1-P_0}E\biggl[\biggl(\frac{1}{K}-X_T\biggr)_+\mathbbm 1_{\{X_T>0\}}\biggr] +\frac{X}{1-P_0}- \frac{1}K\\
&= \frac{1}{1-P_0}\biggl(E\biggl[\biggl(X_T-\frac{1}{K}\biggr)_+\biggr] - X + \frac{1-P_0}{K}\biggr) +\frac{X}{1-P_0}- \frac{1}K\\
&= \frac{1}{1-P_0} E\biggl[\biggl(X_T-\frac{1}K\biggr)_+\biggr].\\
\end{align*} In other words, Calls on \(X^{**}_T\) are still Calls on
\(X_T\), with an apposite rescaling. The main difference between the two
type of Calls is that while Calls on \(X_T\) might have a derivative
larger than \(-1\) in \(0\) because of the non-null probability of
\(X_T\) to nullify, the derived Calls on \(X_T^{**}\) will anyways have
a derivative in \(0\) strictly equal to \(-1\).

	\hypertarget{in-terms-of-pricing-functions}{%
\subsection{In terms of pricing
functions}\label{in-terms-of-pricing-functions}}

	We have introduced the change of probability measure to avoid mass in
zero, now we can consider contracts written in the new probability
spaces. From the previous discussion in
\cref{the-change-of-probability-with-a-mass-at-0} we get the following.

\begin{lemma}

Assume there is a function $C$ of two variables such that $C(X, K) = E[(X_T-K)_+]$. Then
\begin{align*} 
E^*\bigl[\bigl(X_T^*-K\bigr)_+\bigr] &= \frac{K}{X} C\biggl(X,\frac{1}K\biggr)+\frac{1-P_0}{X}- K\\
E^{**}\bigl[\bigl(X_T^{**}-K\bigr)_+\bigr] &= \frac{1}{1-P_0} C(X, K).
\end{align*}

\end{lemma}

	We are interested by necessary and sufficient conditions on \(C\) such
that there exist functions \(C^*\) and \(C^{**}\) satisfying
\begin{equation}\label{eqCstarCstarstar}
\begin{aligned} 
C^*(X^*,K) &= E^*\bigl[\bigl(X^*_T-K\bigr)_+\bigr]\\
C^{**}(X^{**},K) &= E^{**}\bigl[\bigl(X^{**}_T-K\bigr)_+\bigr]
\end{aligned}
\end{equation} where \(X^*:=\frac{1-P_0}{X}\) and
\(X^{**}:=\frac{X}{1-P_0}\).

	The usual situation where \(P_0 = 0\) constantly leads, given that
\(X^*=\frac{1}{X}\) and \(X^{**}=X\), to the formulas \begin{align*}
C^*(X^*,K) &= K X^* C\biggl(\frac{1}{X^*},\frac{1}K\biggr)+X^*- K\\
C^{**}(X^{**},K) &= C(X^{**}, K).
\end{align*} which are described in Theorem 2.2.2. of
\cite{tehranchi2020black}.

	Note that in the degenerate case where \(C(X,K)=(X-(1-P_0)K)_+\), it
holds
\(\frac{K}{X} C\bigl(X,\frac{1}K\bigr)+\frac{1-P_0}{X}- K = (K-X^*)_++X^*-K = (X^{*}-K)_+\)
and \(\frac{1}{1-P_0} C(X, K) = (X^{**}-K)_+\), so that the required
property in \cref{eqCstarCstarstar} holds.

	Going back to the general case, the following lemma finds a sufficient
condition for the existence of the functions \(C^*\) and \(C^{**}\) in
case \(P_0\) is given by some function \(\hat P_0(X)=P_0\).

\begin{lemma}\label{lemmaCondSuffF}

A sufficient condition for the existence of functions $C^*$ and $C^{**}$ satisfying \cref{eqCstarCstarstar} is the existence of a function $f$ such that $X = f\bigl(\frac{X}{1-P_0}\bigr)$ where $P_0=\hat P_0(X)$. Then
\begin{align*}
C^*(X^*,K) &:= \frac{K}{f\bigl(\frac{1}{X^*}\bigr)} C\biggl(f\biggl(\frac{1}{X^*}\biggr),\frac{1}K\biggr)+X^*- K\\
C^{**}(X^{**},K) &:= \frac{X^{**}}{f(X^{**})} C(f(X^{**}), K).
\end{align*}

\end{lemma}

This lemma generalizes the no-mass at \(0\) situation where \(f\) is the
identity function.

	\hypertarget{in-terms-of-normalized-pricing-functions}{%
\subsubsection{In terms of normalized pricing
functions}\label{in-terms-of-normalized-pricing-functions}}

	It is also interesting to look at the normalized partial pricing
functions (for a fixed pair \(X,P_0\)) where the strike is replaced by
the moneyness \(m\), and the price is scaled by the underlier as in
\cref{normalized-call-prices}.

\begin{lemma}

It holds
\begin{align*} 
c^*(m) &= mc\biggl(\frac{1}{-c'(0_+)m}\biggr)+1- m\\
c^{**}(m) &= c\biggl(\frac{m}{-c'(0_+)}\biggr).
\end{align*}
In particular
\begin{align*}
c^{*'}(0_+) &= c(\infty)-1\\
c^{**'}(0_+) &= -1.
\end{align*}

\end{lemma}

\begin{proof}

From definitions, it holds
\begin{align*} 
c^*(m) &= \frac{m}{X} C\biggl(X,\frac{X}{(1-P_0)m}\biggr)+1- m\\
c^{**}(m) &= \frac{1}{X} C\biggl(X, \frac{X}{1-P_0} m\biggr).
\end{align*}
The first conclusion follows from the definition of $c(m) = \frac{C(X, Xm)}{X}$ and from the fact that $P_0 = P(X_T=0)=1+c'(0_+)$.

It is easy to show that the derivative of $c^{**}$ in $0_+$ is $-1$, while the derivative of $c^*$ in $0_+$ is $c(\infty)-1+\lim_{m\to 0_+}\frac{c'\bigl(\frac{1}{-c'(0_+)m}\bigr)}{-c'(0_+)m}$. Now by Fubini, $\frac{d}{dK} C(X,K)=-P(X_T>K)=-E[\mathbbm{1}_{\{X_T>K\}}]$ and it holds
\begin{align*}
X &= E[X_T\mathbbm{1}_{\{X_T>K\}}] + E[X_T\mathbbm{1}_{\{X_T\leq K\}}]\\
&\geq KE[\mathbbm{1}_{\{X_T>K\}}] + E[X_T\mathbbm{1}_{\{X_T\leq K\}}].\\
\end{align*}
Letting $K$ to $\infty$, the term $E[X_T\mathbbm{1}_{\{X_T\leq K\}}]$ goes to $E[X_T] = X$, so that the term $KE[\mathbbm{1}_{\{X_T>K\}}] = -K\frac{d}{dK} C(X,K)$ must go to $0$. As a consequence, $kc'(k)=\frac{K}{X}\frac{d}{dK} C(X,K)$ goes to $0$ as $k=\frac{1}{-c'(0_+)m}$ goes to $\infty$.

\end{proof}

	\hypertarget{the-lifted-calls-on-calls}{%
\subsection{The lifted Calls on Calls}\label{the-lifted-calls-on-calls}}

	Let us go back now to the case of the relative Call on Call pricing
function which satisfies \[C(S, K+M) = \hat C_K\bigl(C(S, K), M\bigr).\]
This means that we set the random variable \(X_T\) in \cref{defPstar} to
be \((S_T-K)_+\), so that \(X_T=0\) iff \(S_T \leq K\).

In this contest, the random variables \(X_T^*\) and \(X_T^{**}\)
correspond to \(\frac{1}{(S_T-K)_+}\) and \((S_T-K)_+\) seen as random
variables in the set \(\{S_T>K\}\), so that \begin{align*}
X &= C(S,K)\\
X^* &= \frac{1-P(S_T\leq K)}{C(S,K)}\\
X^{**} &= \frac{C(S,K)}{1-P(S_T\leq K)}.
\end{align*}

	Then we want to find functions \(\hat C^*_K\) and \(\hat C^{**}_K\) such
that \begin{equation}\label{eqCKstar}
\begin{aligned}
\hat C_K^*(X^*,M) &= \frac{M}{C(S,L)} \hat C_K\biggl(C(S,K),\frac{1}M\biggr)+\frac{1-P(S_T \leq K)}{C(S,L)}- M\\
\hat C^{**}_K(X^{**}, M) &= \frac{1}{1-P(S_T \leq K)}\hat C_K\bigl(C(S, K),M\bigr)
\end{aligned}
\end{equation} and we call \emph{lifted} Calls on Calls the prices
\(\hat C^{**}_K\).

In the following proposition we reconsider \cref{lemmaCondSuffF} and
state a sufficient condition for the existence of such functions
\(\hat C_K^*\) and \(\hat C_K^{**}\).

	\begin{proposition}[Sufficient condition for the existence of $\hat C^*_K$ and $\hat C^{**}_K$]\label{lemmaExistenceF}

A sufficient condition for the existence of functions $\hat C^*_K$ and $\hat C^{**}_K$ satisfying \cref{eqCKstar} is that the function
$$S \to -\frac{d}{dS} C(S,K)\frac{d}{dK} C(S,K) + C(S,K)\frac{d^2}{dSdK} C(S,K)$$
has constant sign for fixed $K$.

In the homogeneous case this is equivalent to the condition that the function
\begin{equation}\label{eqG1Homogeneous}
kc'(k)^2-c(k)c'(k)-kc(k)c''(k)
\end{equation}
has constant sign.

In the Black-Scholes case, the condition holds true.

\end{proposition}

	\begin{proof}

From \cref{lemmaCondSuffF}, a sufficient condition for the existence of $\hat C^*_K$ and $\hat C^{**}_K$ is the existence of a function $f$ such that $f\bigl(\frac{C(S, K)}{1-P(S_T \leq K)}\bigr)=C(S,K)$. This condition is one to one with the fact that the function $g(S)=\frac{C(S, K)}{1-P(S_T \leq K)}$ is monotone. Indeed, if $g(S^1)=g(S^2)$, then $f(g(S^1))=f(g(S^2))$, i.e. $C(S^1,K)=C(S^2,K)$ and $S^1=S^2$ since $C(\cdot,K)$ is a monotone function. On the other hand, if $g$ is monotone, then the function $f(x)=C(g^{-1}(x),K)$ is the required function.

Observe that $1-P(S_T \leq K) = -\frac{d}{dK}C(S,K)$. Then we should prove that $g(S)=\frac{C(S,K)}{-\frac{d}{dK} C(S,K)}$ is monotone. The derivative of $g$ has the sign of
$$-\frac{d}{dS} C(S,K)\frac{d}{dK} C(S,K) + C(S,K)\frac{d^2}{dSdK} C(S,K).$$

In the homogeneous case
$$g(S) = \frac{c(k)S}{-c'(k)}$$
and it is monotone iff \cref{eqG1Homogeneous} has constant sign.

In the Black-Scholes case, we have
\begin{align*}
\text{bs}(k) &= \Phi(d_1)-k\Phi(d_2)\\
\text{bs}'(k) &= -\Phi(d_2)\\
\text{bs}''(k) &= \frac{\phi(d_2)}{kv}
\end{align*}
where $d_{1,2} = -\frac{\log k}{v} \pm \frac{v}2$ and $v=\sigma\sqrt T$. The expression in \cref{eqG1Homogeneous} becomes $\Phi(d_1)\Phi(d_2)-(\Phi(d_1)-k\Phi(d_2))\frac{\phi(d_2)}{v}$ and, using the equality $k\phi(d_2)=\phi(d_1)$, the latter expression becomes $\frac{1}{v}(\phi(d_1)\Phi(d_2) + v \Phi(d_1)\Phi(d_2) - \phi(d_2)\Phi(d_1))$. Since $d_1=d_2+v$, we shall study the sign of $\phi(x+v)\Phi(x) + v\Phi(x+v)\Phi(x) - \phi(x)\Phi(x+v)$ for $v>0$. Dividing by $\phi(x+v)\phi(x)$, this reduces to study
\begin{equation}\label{eqRxy}
R(x) + v R(x+v)R(x) - R(x+v)
\end{equation}
where $R(x)=\frac{\Phi(x)}{\phi(x)}$. Let us consider the function $\mathcal R_x(v) = \frac{R(x+v)}{1+v R(x+v)}$ for a fixed $x$. Its derivative can be simplified to $\frac{R'(x+v)-R^2(x+v)}{(1+v R(x+v))^2}$, which has the sign of $\phi(z)^2+z\Phi(z)n(z)-\Phi(z)^2=\phi(z)^2 \bigl(1+z R(z)-R(z)^2 \bigl)$ where $z=x+v$.

Observe now that $R(z)=r(-z)$ where $r$ denotes the Mill's ratio, whence  $1+z R(z)-R(z)^2=1-tr(t)-r(t)^2$ with $t=-z$. We know that $1-tr(t)=-r'(t)$, and also from Sampford \cite{sampford1953some} that $1 > \frac{1}r'=\frac{-r'}{r^2}$, which proves that $1+z R(z)-R(z)^2 < 0$.

Therefore $\mathcal R_x(v)$ attains its maximum for $v$ going to $0$, where $\mathcal R_x(0)=R(x)$. As a consequence the expression \cref{eqRxy} is always positive and so is \cref{eqG1Homogeneous}.

\end{proof}

	\Cref{lemmaExistenceF} is of particular interest since it shows that in
the Black-Scholes case there exist functions \(\hat C^*_K\) and
\(\hat C^{**}_K\) satisfying \cref{eqCKstar}. In particular, in the
Black-Scholes case it is possible to defined lifted Call on Call pricing
functions, i.e.~Calls on Calls with the property of a derivative equal
to \(-1\) at \(0\). In particular, we have proved that in the
Black-Scholes case the function
\(g(X)=\frac{X}{\Phi(d_2(C^{-1}(X,K),K,v))}\) where
\(d_2(a,b,v)=-\frac{\log\frac{b}{a}}{v}-\frac{v}{2}\) is invertible in
the variable \(X\) and its inverse is \(f(X^{**}):=g^{-1}(X^{**})\)
which satisfies
\[f\biggl(\frac{X}{\Phi(d_2(C^{-1}(X,K),K,v))}\biggr) = X.\] The
function \(f\), however, does not present an explicit form.

Using \cref{propQuasiClosedCPT} we can write the formula for the lifted
Call on Call pricing function in the Black-Scholes case:

	\begin{proposition}[Formula for the lifted Call on Call pricing function in the Black-Scholes model]

For $K,M,v=\sigma\sqrt T\geq 0$, let $\psi(a,v;\Phi) = \inf_{r \in \mathbb R} \frac{a+\Phi(r-v)}{\Phi(r)}$. Call $\tilde f(\cdot,v)$ the inverse of the function
$$z\to\frac{z}{\Phi\bigl(d_2(\psi(z,v;\Phi),1,v)\bigr)}.$$
Then for Black-Scholes prices it holds
$$\hat C_K^{**}(X^{**},M) = \frac{1}{\Phi(d_2(\hat s,1,v))}\Bigl(K\hat s\Phi\bigl(d_1(K\hat s,K+M,v)\bigr) - (K+M)\Phi\bigl(d_2(K\hat s,K+M,v)\bigr)\Bigr)$$
where $d_{1,2}(a,b,v)=-\frac{\log\frac{b}{a}}{v} \pm \frac{v}{2}$ and $\hat s = \psi\bigl(\tilde f\bigl(\frac{X^{**}}K,v\bigr),v;\Phi\bigr)$.

\end{proposition}

	\hypertarget{relation-with-the-transformation-mathbb-u_k}{%
\subsubsection{\texorpdfstring{Relation with the transformation
\(\mathbb U_k\)}{Relation with the transformation \textbackslash mathbb U\_k}}\label{relation-with-the-transformation-mathbb-u_k}}

	Let us consider the lifted Call on Call \(\hat C_K^{**}(X^{**},M)\) and
its normalized function
\[\hat c^{**}_K(m) = \frac{\hat C^{**}_K(X^{**},X^{**}m)}{X^{**}}\]
where \(X^{**} = \frac{X}{1-P(X_T=0)}\). Then in turn it holds
\begin{align*}
\hat c^{**}_K(m) &= \frac{1}{X}\hat C_K\biggr(X, \frac{X}{1-P(X_T=0)}m\biggl)\\
&= \frac{1}{X}C\biggl(C^{-1}(X,K), K + \frac{X}{1-P(X_T=0)}m\biggr).
\end{align*} We write \(k=\frac{K}{C^{-1}(X,K)}\) and define the
function \(c(l):=\frac{C(C^{-1}(X,K),lC^{-1}(X,K))}{C^{-1}(X,K)}\), then
\(X = c(k)C^{-1}(X,K)\) and \begin{align*}
\hat c^{**}_K(m) &= \frac{c\bigl(k+\frac{c(k)}{1-P(X_T=0)}m\bigr)}{c(k)}
\end{align*} where
\(1-P(X_T=0) = -\partial_KC(C^{-1}(X,K),K) = -c'(k)\). Then
\[\hat c^{**}_K(m) = \mathbb U_{k(X,K)}c(m)\] where
\(k(X,K)=\frac{K}{C^{-1}(X,K)}\).

While in \cref{a-transformation-in-the-tehranchi-space} we have proved
that the transformation \(\mathbb T_k\) corresponds to the normalization
of the relative Call on Call \(\hat C_K\), here we showed that the
transformation \(\mathbb U_k\) actually corresponds to the normalization
of the lifted Call on Call \(\hat C_K^{**}\). In other words we showed
the following.

\begin{lemma}

The transformations $\mathbb T_k$ and $\mathbb U_k$ are linked to relative Calls on Calls $\hat C_K$ and lifted Calls on Calls $\hat C_K^{**}$ via
\begin{align*}
\hat C_K(X,M) &= \mathbb T_{k(X,K)}c\Bigl(\frac{M}{X}\Bigr)X\\
\hat C_K^{**}(X^{**},M) &= \mathbb U_{k(X,K)}c\Bigl(\frac{M}{X^{**}}\Bigr)X^{**}
\end{align*}
where $k(X,K)=\frac{K}{C^{-1}(X,K)}$.

\end{lemma}

	\hypertarget{implied-volatility-of-the-relative-pricing-functions}{%
\section{Implied volatility of the relative pricing
functions}\label{implied-volatility-of-the-relative-pricing-functions}}

	The Calls on Calls prices and their lifted prices can be studied from an
implied volatility point of view. It is of extreme interest that, even
in the case of underlying Black-Scholes Calls, the relative prices
actually hide smile shapes which are not constant. Also, we can state
something more on these smiles, as we will see that they always explode
for small strikes while behave as the original smiles at \(\infty\).

We define the implied total volatility \(\hat v_K(X,M)\) of the relative
Call on Call \(\hat C_K(X,M)\) and the implied total volatility
\(\hat v_K^{**}(X^{**},M)\) of the lifted Call on Call
\(\hat C_K^{**}(X^{**},M)\) to be the value of \(v=\sigma\sqrt T\)
satisfying \[\hat C_K(X,M) = \text{BS}(X,M,v)\] and
\[\hat C_K^{**}(X^{**},M) = \text{BS}(X^{**},M,v)\] respectively, where
\(\text{BS}\) is defined in \cref{eqBS}.

	We are interested in the functions \(M\to\hat v_K(X,M)\) and
\(M\to\hat v_K^{**}(X,M)\) and in particular to their asymptotics for
\(M\) going to \(\infty\) and \(0\). We can look at the limits of this
function studying the relations \begin{align*}
\hat C_K(X,M) &= C(C^{-1}(X,K),K+M)\\
\hat C_K^{**}(X^{**},M) &= \frac{X^{**}}{f(X^{**})}\hat C_K(f(X^{**}),M)
\end{align*} where we suppose that there exists \(f\) such that
\(f(X^{**}) = X\).

We denote with \(v(S,L)\) the implied total volatility of the Call
option with strike \(L\) and underlier \(S\), so that we shall study the
relation between \(\hat v_K(X,M)\) (or \(\hat v_K^{**}(X^{**},M)\)) and
\(v(C^{-1}(X,K),K+M))\). We will also denote
\(\hat d_{1,2}(X,M):=d_{1,2}\bigl(X,M,\hat v_K(X,M)\bigr)\) and
\(\hat d_{1,2}^{**}(X^{**},M):=d_{1,2}\bigl(X^{**},M,\hat v_K^{**}(X^{**},M)\bigr)\).
When it is clear, we will drop the dependence to \(X\) and \(X^{**}\).

	\begin{remark}\label{remarkdLee}

In the framework of arbitrage-free prices the functions $d_1$ and $d_2$ are monotone for the results found by Fukasawa in \cite{fukasawa2012normalizing}. The condition of surjectivity is not granted a priori. In general, it always holds that the function $d_1$ goes at $+\infty$ for small strikes and the function $d_2$ goes at $-\infty$ for large strikes. However, the function $d_1$ goes to $-\infty$ for large strikes iff Call prices vanish at $\infty$, and $d_2$ goes to $\infty$ for small strikes iff there is no mass at $0$, i.e. the derivative of Call prices is $-1$ in $0$.

Furthermore, Proposition 2.1 in \cite{mingone2022smiles} shows that if $d_1$ is surjective, then the Lee right wing condition holds: $v(S,K)<\sqrt{2\log\frac{K}{S}}$ for $K$ large enough; while if $d_2$ is surjective, then the Lee left wing condition holds: $v(S,K)<\sqrt{-2\log\frac{K}{S}}$ for $K$ small enough.

\end{remark}

	\begin{remark}\label{remarkLimPhi}

For two functions $f$ and $g$ we write $f\sim g$ iff $\lim_x\frac{f(x)}{g(x)}=1$, where the limit of $x$ will depend on the situation.

Integrating by parts the expression for $1-\Phi(x)$ we find
$$1-\Phi(x) = \frac{\phi(x)}{x}-\int_x^\infty\frac{\phi(t)}{t^2}\,dt = \frac{\phi(x)}{x}-\frac{\phi(x)}{x^3}+\int_x^\infty\frac{\phi(t)}{t^4}\,dt$$
so that for $x>0$
$$\frac{\phi(x)}{x}-\frac{\phi(x)}{x^3}<1-\Phi(x)<\frac{\phi(x)}{x}.$$
Then, for $x$ going to $\infty$, $\Phi(-x)\sim\frac{\phi(x)}{x}$.

\end{remark}

	\hypertarget{calls-on-calls-implied-volatility}{%
\subsection{Calls on Calls' implied
volatility}\label{calls-on-calls-implied-volatility}}

	In this section we prove that the Calls on Calls' smiles behave as the
original smiles at \(\infty\), while at small strikes they always
explode with a rate of \(\sqrt{-2\log\frac{M}{X}}\).

\begin{lemma}[Asymptotic behavior of the Calls on Calls' total implied volatility]\label{lemmaImpliedVolhatC}

The Calls on Calls' total implied volatility $\hat v_K(X,M)$ behaves asymptotically as
\begin{itemize}
\item $\sqrt{-2\log\frac{M}{X}} - \hat d_2(X,0)$ for small strikes;
\item the underlying total implied volatility $v(C^{-1}(X,K),K+M)$ for large strike. If it exists, the limit of $\hat v_K(X,M)$ is equal to the limit of $v(C^{-1}(X,K),K+M)$.
\end{itemize}

\end{lemma}

	\begin{proof}

Firstly, we observe that $-\frac{d}{dM}\hat C_K(X,M) = -\partial_K C(C^{-1}(X,K),K+M)$. In $M=0$, this means that Calls on Calls' prices have a derivative in $0$ equal to $-\partial_K C(C^{-1}(X,K),K)$ which, for $K>0$, is strictly larger than $-1$. Then, the function $\hat d_2(M)$ is not surjective and $\hat v_K\sim\sqrt{-2\log\frac{M}{X}}$. More precisely, we can write $\hat v_K = \gamma\sqrt{-2\log\frac{M}{X}}$ where $\gamma \sim 1$ at $\infty$. Substituting in the expression for $\hat d_2(M)$, we find
\begin{equation}\label{eqD2Epsilon}
\hat d_2(M) = \frac{\sqrt{2}}{2}\sqrt{-\log\frac{M}{X}}\biggl(\frac{1}{\gamma}-\gamma\biggr) \sim \hat d_2(0)
\end{equation}
where $\hat d_2(0)$ is the finite limit of $\hat d_2$ for $M=0$. Observe that $\frac{1}{1-\varepsilon} = 1-\varepsilon+o(\varepsilon)$ so that $\frac{1}{1-\varepsilon} - (1-\varepsilon) = 2\varepsilon + o(\varepsilon)$. Setting $\gamma=1-\epsilon$, we find from \cref{eqD2Epsilon} that it must hold $\varepsilon \sim \frac{\hat d_2(0)}{\sqrt{-2\log\frac{M}{X}}}$, or $\gamma\sim 1-\frac{\hat d_2(0)}{\sqrt{-2\log\frac{M}{X}}}$. All in all, we find $\hat v_K \sim \sqrt{-2\log\frac{M}{X}} - \hat d_2(0)$.

Secondly, from the definition of Calls on Calls' prices, it follows that Calls on Calls vanish at increasing strikes iff the underlying Calls vanish. In particular, if $d_1$ has finite limit, then its implied total volatility must explode at $\infty$, and similarly for $\hat d_1$ and $\hat v_K$. Otherwise, if $d_1$ is surjective, then $\hat d_1$ is surjective and in particular for \cref{remarkdLee} the Lee right wing condition holds, i.e. $v(C^{-1}(X,K), M), \hat v_K(X,M)<\sqrt{2\log\frac{M}{X}}$ for large $M$. We can then write the asymptotics of the Call on Call price using \cref{remarkLimPhi} as
$$C_\text{BS}(X,M,\hat v_K) \sim X\frac{\phi(\hat d_1(M))}{-\hat d_1(M)} - M\frac{\phi(\hat d_2(M))}{-\hat d_2(M)}$$
and since $M\phi(\hat d_2(M))=X\phi(\hat d_1(M))$, developing the expressions for $\hat d_1(M)$ and $\hat d_2(M)$, the right hand side becomes
$$C_\text{BS}(X,M,\hat v_K) \sim X\phi(\hat d_1(M))\frac{\hat v_K^3}{\bigl(\log\frac{M}{X}\bigr)^2-\frac{\hat v_K^4}{4}}.$$
Taking the logarithm in the above expression and looking at the dominating terms, we find
$$\log C_\text{BS}(X,M,\hat v_K) \sim -\frac{\hat d_1(M)^2}{2}.$$
If $\hat v_K\in o\bigl(\sqrt{\log\frac{M}{X}}\bigr)$ (this includes the case where it goes to a finite limit) then 
$$\log C_\text{BS}(X,M,\hat v_K) \sim -\frac{1}{2\hat v_K^2}\biggl(\log\frac{M}{X}\biggr)^2,$$
otherwise, if $\hat v_K\sim \hat c\sqrt{\log\frac{M}{X}}$ with $0<\hat c<\sqrt 2$, then
$$\log C_\text{BS}(X,M,\hat v_K) \sim -\frac{\log\frac{M}{X}}{8\hat c^2}(2-\hat c^2)^2.$$

Similarly, we develop the asymptotics of the logarithm of $C_\text{BS}(C^{-1}(X,K),K+M,v)$ and get that if $v\in o\bigl(\sqrt{\log\frac{K+M}{C^{-1}(X,K)}}\bigr)$ these coincide with
$$\log C_\text{BS}(C^{-1}(X,K),K+M,v) \sim -\frac{1}{2v^2}\biggl(\log\frac{K+M}{C^{-1}(X,K)}\biggr)^2.$$
Then, also the logarithm of Call on Call prices must have a similar behavior, so that since $\log\frac{M}{X}\sim\log\frac{K+M}{C^{-1}(X,K)}$ the only possible case is $\hat v_K\in o\bigl(\sqrt{\log\frac{M}{X}}\bigr)$ and equating the asymptotic behaviors it follows $\hat v_K\sim v$. In particular, if it exists, the limit of $\hat v_K(X,M)$ at $\infty$ is equal to the limit of $v(C^{-1}(X,K),K+M)$.

Otherwise, if $v\sim c\sqrt{\log\frac{K+M}{C^{-1}(X,K)}}$ with $0<c<\sqrt 2$, then
$$\log C_\text{BS}(X,M,\hat v_K) \sim -\frac{\log\frac{K+M}{C^{-1}(X,K)}}{8c^2}(2-c^2)^2.$$
Now in order to have equivalent asymptotic behaviors, then there must be a positive $\hat c<\sqrt{2}$ such that $\hat v_K\sim \hat c\sqrt{\log\frac{M}{X}}$ and $\frac{(2-\hat c^2)^2}{\hat c^2} = \frac{(2-c^2)^2}{c^2}$. Solving, the only positive solution is $\hat c=c$, so that again $\hat v_K\sim v$.

\end{proof}

	\hypertarget{relation-with-the-underlying-implied-volatility}{%
\subsubsection{Relation with the underlying implied
volatility}\label{relation-with-the-underlying-implied-volatility}}

	In \cref{propMonotonicityStrike} we show that, in the homogeneous case,
Calls on Calls' prices are increasing as functions of the relative
underlying strike. Then, \(\hat C_{K_1}(X,M)<\hat C_{K_2}(X,M)\) for
every \(K_1<K_2\), which implies
\(\hat v_{K_1}(X,M)<\hat v_{K_2}(X,M)\). In particular, for \(K_1=0\)
and \(K_2=K\), it holds \(\hat C_0(X,M)=C(X,M)=C_\text{BS}(X,M,v(X,M))\)
and \(\hat C_K(X,M)=C_\text{BS}(X,M,\hat v_K(X,M))\), so that looking at
the implied total volatilities it follows \(v(X,M)<\hat v(X,M)\). This
means that the Calls on Calls's implied total volatility is always
larger than the original implied total volatility for fixed moneyness.

	\hypertarget{lifted-calls-on-calls-implied-volatility}{%
\subsection{Lifted Calls on Calls' implied
volatility}\label{lifted-calls-on-calls-implied-volatility}}

	We now look at the behavior of the smiles of lifted Calls on Calls. In
this case, we expect \(\hat d_2^{**}\) to be surjective since prices
have derivative equal to \(-1\) in \(0\). However, we will show that the
smiles still explode for small strikes, while the behavior at \(\infty\)
is as for the original smile.

\begin{lemma}[Asymptotic behavior of the lifted Calls on Calls' total implied volatility]\label{lemmaImpliedVolhatCstar}

The lifted Calls on Calls' total implied volatility $\hat v^{**}_K(X^{**},M)$ behaves asymptotically as
\begin{itemize}
\item $(2-\sqrt{2})\sqrt{-\log\frac{M}{X^{**}}}$ for small strikes;
\item the underlying total implied volatility $v(C^{-1}(f(X^{**}),K),K+M)$ for large strike. If it exists, the limit of $\hat v^{**}_K(X^{**},M)$ is equal to the limit of $v(C^{-1}(f(X^{**}),K),K+M)$.
\end{itemize}

\end{lemma}

	\begin{proof}

It holds $-\partial_M\hat C_K^{**}(X^{**},M) = -\frac{X^{**}}{f(X^{**})}\partial_M\hat C_K(f(X^{**}),M)$ and $-\partial_M\hat C_K(f(X^{**}),M) = 1-P(f(X^{**})=0)=\frac{f(X^{**})}{X^{**}}$, then lifted Calls on Calls' prices have derivative equal to $-1$ at null strikes. In particular, the Lee left wing condition holds: $\hat v_K^{**}(X^{**},M)<\sqrt{-2\log\frac{M}{X^{**}}}$ for small strikes.

From the Put-Call-Parity, it holds
\begin{equation}\label{eqPStar}
P_\text{BS}^{**}(X^{**},M,\hat v_K^{**}) = \frac{X^{**}}{f(X^{**})} C(C^{-1}(f(X^{**}),K),K+M) - X^{**} + M
\end{equation}
where we dropped the dependence of the implied total variance from the underlier and the strike for notation simplicity. For $M$ going to $0$, since both $\hat d_1^{**}$ and $\hat d_2^{**}$ go to $\infty$, the left hand side behaves as
\begin{align*}
P_\text{BS}^{**}(X^{**},M,\hat v_K^{**}) &\sim -X^{**}\frac{\phi(\hat d_1^{**})}{\hat d_1^{**}} + M\frac{\phi(\hat d_2^{**})}{\hat d_2^{**}}\\
&= X^{**}\phi(\hat d_1^{**})\frac{\hat v_K^{**3}}{\bigl(\log\frac{M}{X^{**}}\bigr)^2-\frac{\hat v_K^{**4}}{4}}.
\end{align*}
As in the proof of \cref{lemmaImpliedVolhatC}, we take the logarithm and consider the dominant terms, so that
$$\log P_\text{BS}^{**}(X^{**},M,\hat v_K^{**}) \sim -\frac{\hat d_1^{**2}}{2}.$$

On the other hand, the right hand side in \cref{eqPStar} is equal to
$$\frac{X^{**}}{f(X^{**})}\bigl(C(C^{-1}(f(X^{**}),K),K+M) - f(X^{**})\bigr) + M$$
and, since $C(C^{-1}(f(X^{**}),K),K)=f(X^{**})$, for $M$ going to $0$ the above expression behaves as
$$\frac{X^{**}}{f(X^{**})}\bigl(\partial_K C(C^{-1}(f(X^{**}),K),K)M + \partial^2_K C(C^{-1}(f(X^{**}),K),K)M^2\bigr) + M.$$
Observe that $\frac{X^{**}}{f(X^{**})}=\frac{1}{-\partial_K C(C^{-1}(f(X^{**}),K),K)}$ so that the above expression reduces to
$$\frac{X^{**}}{f(X^{**})}\partial^2_K C(C^{-1}(f(X^{**}),K),K)M^2$$
where the term multiplying $M^2$ is a positive constant. Then, taking the logarithm, this behaves as $2\log M$.

Now, if $\hat v_K^{**}\in o\bigl(\sqrt{-\log\frac{M}{X^{**}}}\bigr)$ then
$$\log P_\text{BS}^{**}(X^{**},M,\hat v_K^{**}) \sim -\frac{1}{2\hat v_K^{**2}}\biggl(\log\frac{M}{X^{**}}\biggr)^2.$$
and equating this with $2\log M$ we find $\hat v_K^{**} \sim \frac{\sqrt{-\log M}}{2}$, which implies $\hat v_K^{**}\notin o\bigl(\sqrt{-\log\frac{M}{X^{**}}}\bigr)$, so that this solution cannot be accepted. If $\hat v_K^{**}\sim c^{**}\sqrt{-\log\frac{M}{X^{**}}}$ with $c^{**}<\sqrt{2}$ then
$$\log P_\text{BS}^{**}(X^{**},M,\hat v_K^{**}) \sim \frac{\log\frac{M}{X^{**}}}{8c^{**2}}(2+c^{**2})^2$$
and equating with $2\log M$ we find that the only admissible solution is $c^{**}=2-\sqrt{2}$.

Regarding the limit for large strikes, the definition of $\hat C^{**}_K$ implies that lifted Calls on Calls vanish at $\infty$ iff the underlying Calls do. Then, $d_1^{**}$ is surjective iff $d_1$ is. So that if $d_1$ has a finite limit at $\infty$, then $\hat v^{**}_K$ must explode. Otherwise, as in the relative Calls on Calls' case, the Lee right wing condition must hold: $\hat v^{**}_K(X^{**},M)<\sqrt{2\log\frac{M}{X^{**}}}$ for large enough $M$. Similarly to the previous section, considering the relation
$$C_{\text{BS}}(X^{**},M,\hat v_K^{**}) = \frac{X^{**}}{f(X^{**})}C_{\text{BS}}(C^{-1}(f(X^{**}),K),K+M,v)$$
and developing for large $M$, we obtain that $\hat v_K^{**}(X^{**},M)$ at $\infty$ behaves as $v(C^{-1}(f(X^{**}),K),K+M)$. In particular, if the limit exists, this is equal to the limit of $v(C^{-1}(f(X^{**}),K),K+M)$ at $\infty$.

\end{proof}

	\hypertarget{examples}{%
\subsection{Examples}\label{examples}}

	\hypertarget{black-scholes-calls-on-calls-implied-volatility}{%
\subsubsection{Black-Scholes Calls on Calls' implied
volatility}\label{black-scholes-calls-on-calls-implied-volatility}}

	In the Black-Scholes case, the implied total volatility is constant. In
\cref{lemmaImpliedVolhatC,lemmaImpliedVolhatCstar} we showed that both
the Calls on Calls' implied total volatility and the lifted Calls on
Calls' implied total volatility explode for \(M=0\) and go to the
original Black-Scholes total implied volatility at \(M=\infty\).

On the left of \cref{figureBSSigmahatstar} we plot the total implied
volatilities \(\hat v_K(X,M)\) and \(\hat v_K^{**}(X,M)\) as functions
of \(M\). We take \(X=\text{BS}(S,K,v)\) where \(S=100\), \(K=110\) and
\(v=0.2\).

	\hypertarget{svi-calls-on-calls-implied-volatility}{%
\subsubsection{SVI Calls on Calls' implied
volatility}\label{svi-calls-on-calls-implied-volatility}}

	We consider now the SVI model as in \cref{the-enriched-svi-models}. For
\(K\) going to \(0\), the corresponding implied total volatility
\(\sqrt{\text{SVI}(\log\frac{K}{S})}\) behaves as
\(\sqrt{b(1-\rho)|\log\frac{K}{S}|}\). Except for a constant, this is
equivalent to the behavior of \(\hat v_K(X,M)\) and
\(\hat v_K(X^{**},M)\) in \(0\). Also, for \(M\) going to \(\infty\),
the three total implied volatilities behave similarly, exploding with a
speed of \(\sqrt{b(1+\rho)\log\frac{K+M}{S}}\).

We show on the right of \cref{figureBSSigmahatstar} the total implied
volatilities \(\hat v_K(X,M)\) and \(\hat v_K^{**}(X^{**},M)\) as
functions of \(M\). We set
\(X=\text{BS}(S,K,\text{SVI}(\log\frac{K}{S}))\) where \(S=100\),
\(K=110\) and the SVI parameters \(a=0.01\), \(b=0.1\), \(\rho=-0.6\),
\(m=-0.5\), \(\sigma=0.1\) are taken as in
\cref{the-enriched-svi-models} in order to guarantee arbitrage-free
prices.

\begin{figure}
	\centering
	\includegraphics[width=.8\linewidth]{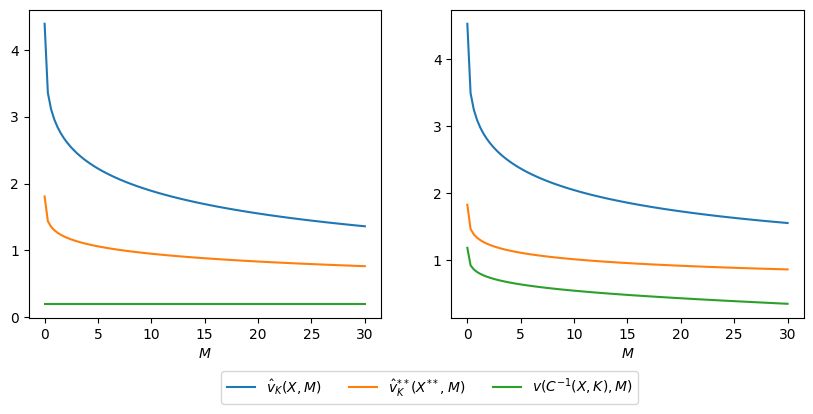}
	\caption{The hidden smiles $\hat v_k(X,M)$ and $\hat v_k^{**}(X^{**},M)$ of the Black-Scholes model (left) and the smiles obtained from SVI prices (right). The current underlier has value $100$, the original strike is set at $110$, the original Black-Scholes implied total volatility at $0.2$, and the original SVI parameters at $a=0.01$, $b=0.1$, $\rho=-0.6$, $m=-0.5$, $\sigma=0.1$.}
	\label{figureBSSigmahatstar}
\end{figure}


\newpage \bibliography{biblio}
\bibliographystyle{plain}

\end{document}